\documentclass[11pt,letterpaper]{article}

\usepackage{typearea}
\typearea{14}
\usepackage[font={small,it}]{caption}

\usepackage{setspace}

\usepackage{amsmath,amsfonts,dsfont,amsthm,amssymb}
\usepackage{algorithmic}
\usepackage[linesnumbered, lined, ruled]{algorithm2e}
\usepackage{enumerate}
\usepackage{graphicx}
\usepackage[dvipsnames]{xcolor}
\usepackage{hyperref}

\usepackage[nameinlink]{cleveref}
\newtheorem{theorem}{Theorem}[section]
\newtheorem{lemma}[theorem]{Lemma}
\newtheorem{fact}[theorem]{Fact}
\newtheorem{corollary}[theorem]{Corollary}
\newtheorem{claim}[theorem]{Claim}
\theoremstyle{definition}
\newtheorem{define}[theorem]{Definition}

\Crefname{fact}{Fact}{Facts}

\newcommand{\tO}{\widetilde{O}}
\newcommand{\tG}{\widetilde{G}}

\newcommand{\tn}{\widetilde{n}}
\newcommand{\tm}{\widetilde{m}}
\newcommand{\ts}{\widetilde{s}}
\newcommand{\tlambda}{\widetilde{\lambda}}
\newcommand{\poly}{\mathrm{poly}}
\newcommand{\E}{\mathbb{E}}
\newcommand{\var}{\mathrm{Var}}
\newcommand{\eps}{\varepsilon}
\newcommand{\eat}[1]{}
\newcommand{\ugp}{u_G(p)}
\newcommand{\Elarge}{E_{\mathrm{large}}}
\newcommand{\Glarge}{G_{\mathrm{large}}}
\newcommand{\Esmall}{E_{\mathrm{small}}}
\newcommand{\Gsmall}{G_{\mathrm{small}}}
\newcommand{\rate}{\mathrm{rate}}
\newcommand{\diff}{\mathrm{d}}

\title{Hypergraph Unreliability in Quasi-Polynomial Time}

\author{Ruoxu Cen\thanks{Department of Computer Science, Duke University. Email: {\tt ruoxu.cen@duke.edu}}
\and Jason Li\thanks{Computer Science Department, Carnegie Mellon University. Email: {\tt jmli@cs.cmu.edu}}
\and Debmalya Panigrahi\thanks{Department of Computer Science, Duke University. Email: {\tt debmalya@cs.duke.edu}}}
\date{}

\begin{document}

\maketitle

\pagenumbering{gobble}

\begin{abstract}
The hypergraph unreliability problem asks for the probability 
that a hypergraph gets disconnected when every hyperedge fails
independently with a given probability. 
For graphs, the unreliability problem has been 
studied over many decades, and multiple fully 
polynomial-time approximation schemes
are known starting with the work of Karger (STOC 1995).
In contrast, prior to this work, no non-trivial result
was known for hypergraphs (of arbitrary rank).

In this paper, we give quasi-polynomial time 
approximation schemes for the hypergraph 
unreliability problem. 
For any fixed $\varepsilon \in (0, 1)$, 
we first give a $(1+\varepsilon)$-approximation 
algorithm that runs in $m^{O(\log n)}$ time on 
an $m$-hyperedge, $n$-vertex hypergraph. Then, 
we improve the running time to 
$m\cdot n^{O(\log^2 n)}$ with an additional
exponentially small additive term in the approximation.
\end{abstract}

\clearpage

\pagenumbering{arabic}

\section{Introduction}
\label{sec:intro}
In the hypergraph unreliability problem, we are given an unweighted hypergraph $G=(V, E)$ and a failure probability $0 < p < 1$. The goal is to compute the probability that the hypergraph disconnects\footnote{A hypergraph is said to disconnect due to the failure of a subset of hyperedges when there is a bi-partition of the vertices such that every surviving hyperedge is entirely contained on either side of the bi-partition. Equivalently, the failed hyperedges must contain all hyperedges in some cut of the hypergraph.} when every hyperedge is independently deleted with probability $p$. The probability of disconnection is called the unreliability of the hypergraph $G$ and is denoted $\ugp$. The hypergraph unreliability problem is a natural generalization of network unreliability which is identically defined but on graphs (i.e., hypergraphs of rank\footnote{The {\em rank} of a hypergraph is the maximum rank of any hyperedge in it, where the rank of a hyperedge is the number of vertices in it.} 2). The latter is a classical problem in the graph algorithms literature that was shown to be \#{\tt P}-hard by Valiant~\cite{Valiant79} and its algorithmic study dates back to at least the 1980s~\cite{karp1989monte,AlonFW95}. By now, several fully polynomial-time approximation schemes achieving a $(1+\eps)$-approximation are known for the network unreliability problem~\cite{Karger99,HarrisS18,Karger16,Karger17,Karger20,CenHLP24}. In contrast, to the best of our knowledge, no non-trivial approximation was known for the unreliability problem on hypergraphs of arbitrary rank prior to this work.

Reliability problems are at the heart of analyzing the robustness of networks to random failures. (This can be contrasted with minimum cut problems that analyze the robustness to {\em worst-case} failures.) Since real world networks often exhibit random failures, there is much practical interest in reliability algorithms with entire books devoted to the topic~\cite{chaturvedi2016network,colbourn1987combinatorics}. However, many basic questions remain unanswered from a theoretical perspective. One bright spot from a theoretical standpoint is the network unreliability problem, for which the first FPTAS was given by Karger in STOC 1995~\cite{Karger99}. Since then, many other FPTAS have been discovered with ever-improving running times~\cite{HarrisS18,Karger16,Karger17,Karger20,CenHLP24}, the current record being a recent $\tO(m+n^{1.5})$-time algorithm (for a fixed $\eps$) due to Cen~{\em et al.}~\cite{CenHLP24}. (Throughout the paper, $m$ and $n$ respectively denote the number of (hyper)edges and vertices in the (hyper)graph.) At the heart of these algorithms is the well-known fact that a graph has a polynomial number of near-minimum cuts -- cuts whose value exceeds that of the minimum cut by at most a constant factor~\cite{Karger93}. This polynomial bound extends to hypergraphs of rank at most $O(\log n)$~\cite{KoganK15} and as a result, the FPTAS for network unreliability also apply to such hypergraphs. However, this approach fails for hypergraphs of arbitrary rank. In general, a hypergraph of rank $r$ can have as many as $\Omega(m\cdot 2^r)$ near-minimum cuts (see Kogan and Krauthgamer~\cite{KoganK15} for an example), which rules out an enumeration of the near-minimum cuts in polynomial time for hypergraphs of large rank. This presents the main technical challenge in obtaining an approximation algorithm for hypergraph unreliability, and the main barrier that we overcome in this paper.

In addition to being a natural and well-studied generalization of graphs in the combinatorics literature, hypergraphs have also gained prominence in recent years as a modeling tool for real world networks. While graphs are traditionally used to model networks with point-to-point connections, more complex ``higher-order'' interactions in modern networks are better captured by hypergraphs as observed by many authors in different domains (see, e.g., the many examples in the recent survey of higher order networks by Bick~{\em et al.} \cite{BickGHS23}). Indeed, the use of random hypergraphs as a modeling tool for real world phenomena has also been observed previously~\cite{GhoshalZCN09}. Therefore, we believe that the study of reliability in hypergraphs is a natural tool for understanding the connectivity properties of such real world networks subject to random failures. We initiate this line of research in this paper and hope that this will be further expanded in the future.

\subsection{Our Results}
We give two algorithms for hypergraph unreliability. The first algorithm is simpler and achieves the following result:
\begin{theorem}\label{thm:first}
    For any fixed $\eps\in (0, 1)$, there is a randomized Monte Carlo algorithm for the hypergraph unreliability problem that runs in $m^{O(\log n)}$ time on an $m$-hyperedge, $n$-vertex hypergraph and returns an estimator $X$ that satisfies $X \in (1\pm\eps)\ugp$ whp.\footnote{whp = with high probability. Throughout the paper, we say that a property holds with high probability if it fails with probability bounded by an inverse polynomial in $n$.}
\end{theorem}
The running time of the algorithm in the theorem above (and also that in the next theorem) is inversely polynomial in the accuracy parameter $\eps$. For brevity, we assume that $\eps$ is fixed throughout the paper and do not explicitly state this dependence in our running time bounds. 

Note that the number of hyperedges in a hypergraph can be exponential in $n$. This makes a quasi-polynomial-time hypergraph algorithm that has a running time of $\poly(m)\cdot n^{\poly\log n}$ qualitatively superior to one that has a running time of $m^{\poly\log(n)}$. (Contrast this to graphs where the two bounds are qualitatively equivalent because $m = O(n^2)$.) To this end, we give a second (more involved) algorithm that achieves this sharper bound on the running time incurring a small additive error in the approximation guarantee.

\begin{theorem}\label{thm:second}
    For any fixed $\eps\in (0, 1)$ and any $\delta \in (0, 1)$, there is a randomized Monte Carlo algorithm for the hypergraph unreliability problem that runs in $m\cdot n^{O(\log n\cdot \log\log (1/\delta))}$ time on an $m$-hyperedge, $n$-vertex hypergraph and returns an estimator $X$ that satisfies $X\in (1\pm\eps)\ugp \pm \delta$ whp.
\end{theorem}

To interpret this result, set $\delta = \exp(-n)$. Then, we get an algorithm that runs in $m\cdot n^{O(\log^2 n)}$ time and returns an estimator $X$ that satisfies $X\in (1+\eps) \ugp \pm \exp(-n)$ whp. In other words, we obtain the sharper running time bound that we were hoping for in exchange for an exponentially small additive error in the approximation. We may also note that in general, a simple Monte Carlo simulation of the hypergraph disconnection event also gives an estimator for $\ugp$ with an additive error. But, this additive error would be exponentially larger than the one in \Cref{thm:second}; in particular, in order to ensure that $X\in (1\pm\eps)\ugp + \exp(-n)$ whp, we would need to run the Monte Carlo simulation $\exp(n)$ times, thereby giving an exponential time algorithm as against the quasi-polynomial running time in \Cref{thm:second}.

\subsection{Our Techniques}
We now give a description of the main technical ideas that are used in our algorithms. Let us start with a rough (polynomial) approximation to $\ugp$. In graphs, this is easy. Let $\lambda$ denote the value of a minimum cut. Since there is at least 1 and at most $O(n^2)$ minimum cuts~\cite{DinitzKL76}, their collective contribution to $\ugp$ is between $p^\lambda$ and $O(n^2)\cdot p^\lambda$. Now, since the number of cuts of value $\le \alpha \lambda$ is at most $n^{O(\alpha)}$~\cite{Karger93}, the collective contribution of {\em all} other cuts to $\ugp$
is also at most $O(n^2)\cdot p^\lambda$ (for sufficiently small $p$, 
else we can just use Monte Carlo sampling). The bound of $O(n^2)$ on the number of minimum cuts continues to hold in hypergraphs
(see \cite{ghaffari2017random,chekuri2018minimum}; this is implicitly shown in \cite{cunningham1983decomposition}).
So, their collective contribution is still between $p^\lambda$ and $O(n^2)\cdot p^\lambda$. But, the number of cuts of value $\le\alpha\lambda$ can be exponential in the rank $r$, and therefore exponential in $n$ for $r = \Omega(n)$~\cite{KoganK15}. Therefore, a na\"ive union bound over these cuts only gives a trivial exponential approximation to $\ugp$.

Our first technical contribution is to show that somewhat surprisingly, the upper bound of $O(n^2)\cdot p^\lambda$ on the value of $\ugp$ continues to hold for hypergraphs of arbitrary rank. As described above, we can't simply use a union bound over cuts, but must go deeper into the interactions between different cuts. To this end, we consider an alternative view\footnote{See \cite{Karger20} for a different use of this alternative view.} of the random failure of hyperedges. For each hyperedge, we generate an independent exponential variable (at unit rate) and superpose the corresponding Poisson processes on a single timeline. We contract each hyperedge as it appears on this timeline; then, the disconnection event corresponds to having $\ge 2$ vertices in the contracted hypergraph at time $\ln (1/p)$. As hyperedges contract, the vertices (which we call \emph{supervertices}) of the contracted hypergraph represent a partition of the vertices of the original hypergraph; we assign leaders to the subsets in this partition in a way that we can argue that any two vertices survive as leaders till time $\ln (1/p)$ with probability at most $p^\lambda$. This allows us to recover the $O(n^2)\cdot p^\lambda$ bound on the value of $\ugp$ by a union bound on all vertex pairs.

We now use this rough $O(n^2)$ approximation to $\ugp$ in designing a recursive algorithm. We generate a random hypergraph $H$ by contracting hyperedges in $G$ with probability $1-q$ for some $q > p$. (See \cite{Karger16,Karger17,Karger20,CenHLP24} for the use of random contraction in network unreliability.) The intuition is that by coupling, these edges will survive if the failure probability is $p$; hence, contracting them does not affect the disconnection event. The algorithm now makes a recursive call on $H$ with the conditional failure probability $p/q$ and obtains an estimator for $u_H(p/q)$. But, how do we bound the variance due to the randomness of $H$? This is where the $O(n^2)$-approximation comes in handy -- it bounds the range of $u_H(p/q)$ to $O(n^2)\cdot (p/q)^\lambda$, thereby giving a bound of $n^2 \cdot q^{-\lambda}$ on the (relative) variance of the overall estimator.\footnote{The relative variance of a random variable $X$ is defined as $\var[X]/\E^2[X]$.} Thus, if we select $q$ such that $q^{-\lambda} = \poly(n)$, then we only need a polynomial number of random trials. 

For this plan to work, we need to we make progress in the recursion, i.e., make recursive calls on subgraphs that are smaller by a constant factor. Unfortunately, we are unable to ensure this in hypergraphs of arbitrarily large rank. To see this, consider a hypergraph containing $n$ hyperedges of rank $n-1$, i.e., $\lambda = n-1$. In this case, we have $n^2\cdot q^{-n+1}$ trials and the probability of each trial returning the input hypergraph is $q^n$ (if none of the $n$ hyperedges is contracted). So, $\ge 1$ recursive calls (in expectation) will run on the input hypergraph itself, which defeats the recursion. However, we show that this is really an extreme scenario and we can make sufficient progress in all hypergraphs {\em with rank at most $n/2$} -- we call these {\em universally small} and the rest {\em existentially large} hypergraphs.

We are now left to handle existentially large hypergraphs. This is where the two algorithms (\Cref{thm:first} and \Cref{thm:second}) differ. The first algorithm  (\Cref{thm:first}) simply enumerates over all outcomes (survival/failure) of the large hyperedges, i.e., those of rank $> n/2$. To do this efficiently, it orders the large hyperedges and creates a new recursive instance based on the first large hyperedge that is contracted in this order. This generates $\ell \le m$ subproblems, where $\ell$ denotes the number of large hyperedges. In the last subproblem, all the $\ell$ large hyperedges fail (i.e., none of them is contracted) and we are left with a universally small hypergraph. In all the other subproblems, at least one large hyperedge is contracted and we are left with a hypergraph containing at most $n/2$ vertices. So, we make progress in either case. 

The second algorithm (\Cref{thm:second}) cannot afford to enumerate over all large hyperedges. Instead, it partitions the set of hyperedges in $G$ into the large and small hyperedges and creates two hypergraphs, $\Glarge$ and $\Gsmall$. Now, for $G$ to be disconnected, both $\Gsmall$ and $\Glarge$ must be disconnected (but not vice-versa!). Recall that earlier, we ran into a problem where our na\"ive sampling process could not make progress in terms of reducing the size of the hypergraph when sampling large hyperedges. This was epitomized by a hypergraph containing $n$ hyperedges of rank $n-1$ each. But, if we think of this instance in isolation, then it is actually quite easy to estimate $\ugp$ in this hypergraph. This is because whenever the hypergraph disconnects, it does so at a degree cut\footnote{A degree cut is a cut that separates one vertex from the rest of hypergraph.}
of a vertex. So, there are only $n$ cuts that we need to enumerate over. In fact, this property is true for the hypergraph $\Glarge$ obtained from any hypergraph $G$; since every pair of large hyperedges share at least one vertex, any disconnected sub-hypergraph must have an isolated vertex. We exploit this property by writing a DNF formula for all the degree cuts of $\Glarge$ (where each variable denotes survival/failure of a large hyperedge) and use the classical importance sampling technique of Karp, Luby, and Madras~\cite{karp1989monte} to generate a sample of $\Glarge$ conditioned on it being disconnected. 

How do we augment this sample in $\Gsmall$? We have two cases. To understand the distinction, let us informally imagine that the minimum cuts of $\Glarge$ and $\Gsmall$ coincide, and they form the minimum cut of $G$. (Of course, this is not true in general!) The two cases are defined based on the relative values of the minimum cuts in $\Glarge$ and $\Gsmall$. If $\Glarge$ contributes most of the hyperedges to the mincut (we call this the {\em full revelation} case), then the probability that $\Gsmall$ gets disconnected is quite high (recall that $\ugp \ge p^\lambda$). In this case, it suffices to do Monte Carlo sampling in $\Gsmall$ to augment the sample obtained from $\Glarge$. The other case is when $\Gsmall$ contributes a sizeable number of hyperedges to the minimum cut (we call this the {\em partial revelation} case). Note that the extreme example of this second case is when $\Glarge$ is empty, i.e., when the hypergraph is universally small. This suggests generalizing the use of random contraction from universally small hypergraphs to this case, i.e., failing hyperedges at a higher probability of $q > p$ in a recursive step. But, to synchronize the sample across $\Glarge$ and $\Gsmall$, we must use the same value $q$ in $\Glarge$ as well. Unfortunately, as we observed earlier, the algorithm might not make progress in terms of the size of the hypergraph in this case. To overcome this, we introduce a second recursive parameter, that of the value of the failure probability itself. This second recursive parameter now requires us to define a new base case when the probability of failure is very small (denote the threshold by a parameter $\delta$) -- this is where we incur the additive loss of $\delta$ in the approximation. The overall running time is now given by the fact that each subproblem branches into polynomially many subproblems, and the depth of the recursion is bounded by $\log n \log\log (1/\delta)$ where the first term comes from the recursion on size and the second term from that on the failure probability.

\paragraph{Organization.}
We give some preliminary definitions and terminology in \Cref{sec:prelim}. We then establish \Cref{thm:first} in \Cref{sec:enumeration}. Finally, we establish \Cref{thm:second} in \Cref{sec:dnf-sampling}. We give some concluding thoughts in \Cref{sec:conclusion}.

\section{Preliminaries}
\label{sec:prelim}
\paragraph{Hypergraphs.}
We start with some basic notations for hypergraphs. A hypergraph $G = (V, E)$ comprises a set of vertices $V$ and set of hyperedges $E$, where each hyperedge $e\in E$ is a non-empty subset of the vertices, i.e., $\emptyset \subset e\subseteq V$. The {\em rank} of a hyperedge $e$ is $|e|$; the rank of a hypergraph $G$, denoted $r_G$, is the maximum rank of a hyperedge in $G$. 

For any hypergraph $G=(V, E)$ and subset of hyperedges $F\subseteq E$, denote $G-F := (V, E\setminus F)$ to be the hypergraph after deleting the hyperedges in $F$ from $G$. 
A \emph{cut} in a hypergraph is defined as a set of edges $C$ such that $G-C$ is disconnected. The value of a cut $C$ is the number of hyperedges in $C$. A \emph{minimum cut} of a hypergraph is a cut of minimum value. We denote the value of a minimum cut in a hypergraph $G$ by $\lambda_G$. The following is a known result (follows from Theorem 4 in \cite{chekuri2021isolating} using the maximum flow algorithm in \cite{chen2022maximum}):
\begin{theorem}\label{thm:mincut-CQ}
    The minimum cut of a hypergraph can be computed in $\left(\sum_e |e|\right)^{1+o(1)}$ time.
\end{theorem}

In this paper, we often make use of hyperedge {\em contractions}. Contracting a hyperedge $e$ in a hypergraph $G$ replaces the vertices in $e$ by a single vertex to form a new hypergraph denoted $G/e := (V/e, E/e)$. Note that there is a natural surjective map $\phi:V\to V/e$ that maps vertices in $e$ to the contracted supervertex in $V/e$, and maps vertices outside $e$ to themselves.
Each hyperedge $e\in E$ is replaced in $E/e$ by an element-wise mapped set $\{u\in e:\phi(u)\}$.
By extension, contracting a {\em set} of hyperedges $F=\{e_1, e_2, \ldots\}$ is equivalent to contracting all hyperedges in $F$ in arbitrary order: we write $G/F := (((G/e_1)/e_2)\ldots)/ e_k$.
$H$ is called a contracted hypergraph of $G=(V,E)$ if $H=G/F$ for some $F\subseteq E$.
For distinction between the uncontracted vertices in $G$ and the contracted vertices in $H$, we usually call the former {\em vertices} and the latter {\em supervertices}.

A key operation in our algorithm is uniform random hyperedge contraction.
We use $H\sim G(q)$ for some $q\in (0, 1)$ to denote the distribution of a random contracted hypergraph $H$ obtained from $G$ by contracting each hyperedge independently with probability $1-q$.
The next lemma 
states that $u_H(p/q)$ is an unbiased estimator of $\ugp$:
\begin{lemma}\label{lem:rc-unbias}
    Suppose $H\sim G(q)$ and $q\ge p$. Then, $\E[u_H(p/q)] = \ugp$.
\end{lemma}
\begin{proof}
    Deleting each hyperedge independently with probability $p$ is equivalent to first choosing each hyperedge with probability $q\ge p$ and then deleting each chosen edge with probability $p/q$. 
    $\ugp$ is the probability that $G$ disconnects in the former distribution. In the latter distribution, note that the hyperedges that are not chosen must be connected in the resulting hypergraph after deletion, so contracting them does not affect the disconnection event of the resulting hypergraph.  Thus, $\E[u_H(p/q)]$ is the probability that $G$ disconnects in the latter distribution.
    In conclusion, the two probabilities are equal.
\end{proof}

\paragraph{Random Variables.} Next, we give some basic facts about random variables that we will use in this paper. 
All random variables considered in the paper are non-negative.


The {\em relative variance} of a random variable $X$ is
    \[\eta[X]=\frac{\var[X]}{(\E[X])^2}
    =\frac{\E[X^2]}{(\E[X])^2}-1.\]

Since we use a biased estimator in \Cref{thm:second}, we need a non-standard (capped) version of relative variance. We define it and state its properties below.
\begin{define}[Capped relative variance]
The \emph{($\delta$-)capped relative variance} of random variable $X$ is
    \[\eta_\delta[X] = \frac{\var[X]}{\max\{(\E[X])^2, \delta^2\}}.\]
\end{define}

We state some basic facts about capped relative variance (proofs in \Cref{sec:proofs}).
Note that relative variance is a special case of capped relative variance when $\delta=0$. Therefore, these facts also hold for relative variance as a special case.

\begin{fact}\label{fact:cap-relvar-average}
The average of $M$ independent samples of $X$ has capped relative variance $\frac{\eta_\delta[X]}{M}$.
\end{fact}

\begin{fact}\label{lem:cap-relvar-concatenate}
Suppose $Y$ is an unbiased estimator of $x$, and conditioned on a fixed $Y$, $Z$ is a biased estimator of $Y$ with bias in $[-\delta, 0]$ and capped relative variance $\eta_\delta[Z|Y] \le h$. Then $$\eta_\delta[Z] \le 4\cdot (\eta[Y]+1)\cdot (h+1).$$

In particular, when $\delta=0$ (i.e.\ for relative variance of unbiased estimator $Z$), there is a stronger bound
$$\eta[Z] \le (\eta[Y]+1)\cdot (h+1)-1.$$
\end{fact}

\begin{fact}\label{fact:cap-relvar-mult}
    Suppose $X$ and $Z$ are independent random variables with expectation in $(0, 1)$, and $\delta \in [0,1]$. Then 
    \[\eta_\delta[XZ] \le \eta_\delta[X]\cdot \eta_\delta[Z] + \eta_\delta[X] + \eta_\delta[Z].\]
\end{fact}

\begin{lemma}\label{lem:cap-relvar-approx}
    The median-of-average of $\frac{\eta_\delta[X]}{\eps^2}$ independent samples of $X$ is a $(1\pm \eps, \delta)$-approximation of $\E[X]$.
\end{lemma}

The next two facts are for relative variance and proved in \Cref{sec:proofs}:
\begin{fact}\label{fact:relvar-linear-combination}
    If $X$ is a convex combination of independent non-negative random variables $X_1, \ldots, X_k$, i.e., $X=\sum_{i\le k}\alpha_iX_i$ for $\alpha_i\ge 0$ and $\sum_{i\le k}\alpha_i = 1$, 
    then $\eta[X] \le \max_{i\le k}\eta[X_i]$.
\end{fact}

\begin{fact}\label{fact:relvar-max}
    If a non-negative random variable $X$ is upper bounded by $M$, then $\eta[X] \le \frac{M}{\E[X]} -1$.
\end{fact}

\paragraph{Exponential distribution.}
Recall that the exponential distribution of rate $r$ gives a continuous random variable $X\ge0$ satisfying $\Pr[X\ge t]=e^{-rt}$ for all $t\ge0$. We state some standard properties of the exponential variables:

\begin{fact}[Moment generating function]\label{fact:exp-mgf}
Let $X$ follow exponential distribution of rate $r$. Then for any $t<r$,
$\E[e^{tX}] =  1/(1-\frac{t}{r})$.
\end{fact}
\begin{fact}[Memoryless property]\label{fact:exp-memoryless}
Let $X$ follow exponentail distribution. Then for any $s, t\ge 0$, $\Pr[X>s+t | X>s] = \Pr[X>t]$.
\end{fact}
\begin{fact}\label{fact:exp-min}
    Let $X_1, X_2, \ldots, X_k$ be independent random variables with exponential distribution of rate $r$, and $X=\min_{i\le k}\{X_i\}$.
    Then, $X$ follows exponential distribution of rate $kr$. Moreover, $X = X_i$ for every value of $i$ with probability $1/k$.
\end{fact}

\paragraph{Monte Carlo sampling.}
Suppose we want to estimate the probability $p_D$ that an event $D$ happens. (For $\ugp$, $D$ is the event that the hypergraph disconnects.) 
The Monte Carlo sampling algorithm first draws a sample from the underlying space. (For $\ugp$, it deletes each hyperedge independently with probability $p$.) The estimator returns 1 if $D$ happens, and 0 otherwise. The following is a standard property of this estimator (proof in \Cref{sec:proofs}):

\begin{lemma}\label{lem:mc-relvar}
    Monte Carlo sampling outputs  an unbiased estimator of $p_D$ with relative variance at most $\frac{1}{p_D}$ and $\delta$-capped relative variance at most $\min\{\frac{1}{p_D}, \frac{1}{\delta}\}$.
\end{lemma} 
Given \Cref{lem:mc-relvar}, we can use \Cref{lem:cap-relvar-approx} to obtain the following:
\begin{lemma}\label{lem:mc-algo}
 We can obtain a $(1+\eps)$-approximation of $p_D$ whp via $O\left(\frac{\log n}{\eps^2\cdot p_D}\right)$ Monte Carlo samples and a $(1+\eps, \delta)$-approximation whp via $O\left(\frac{\log n}{\eps^2\cdot \delta}\right)$ Monte Carlo samples.
\end{lemma}

\paragraph{DNF probability.}
In the \emph{DNF probability} problem, 
we are given a DNF formula $F$ with $N$ variables and $M$ clauses and a value $p\in (0, 1)$. The goal is to estimate the probability $u_F(p)$ that $F$ is satisfied when each variable is \textsc{True} with probability $p$ independently.
This problem is \#{\tt P}-hard even in the special case of $p=\frac 12$ \cite{Valiant79}.
In a seminal work, Karp, Luby and Madras \cite{karp1989monte} provided an FPRAS in $\tO(NM)$ time.

\begin{theorem}[\cite{karp1989monte}]\label{thm:dnf-klm}
    The DNF probability problem can be $(1\pm \eps)$-approximated with success probability $1-\delta$ in $O(NM\ln(1/\delta)/\eps^2)$ time.
\end{theorem}
Our algorithm will need an unbiased estimator for DNF probability. The estimator in \Cref{thm:dnf-klm} could be biased, but we can get an unbiased estimator by using its primitive version, at the cost of a slower running time. We state this in the next two lemmas; these are essentially shown in \cite{karp1989monte}, but we include a proof in the appendix for completeness.
\begin{lemma}\label{lem:dnf-unbias}
    An unbiased estimator of $u_F(p)$ with relative variance at most $1$ can be computed in time $O(NM^2)$.
\end{lemma}

\begin{lemma}[DNF sampling]\label{lem:dnf-sampling}
    There exists an algorithm that draws a sample of values in time $O(NM^2)$ according to the following distribution:
    Each variable independently takes value \textsc{True} with probability $p$ and \textsc{False} with probability $1-p$, conditioned on the fact that the values satisfy $F$.
\end{lemma}


\section{Random Contraction with Large Edge Enumeration}
\label{sec:enumeration}
In this section, we design an $m^{O(\log n)}$-time algorithm that outputs an unbiased estimator of $\ugp$ with relative variance $O(1)$.
It follows by \Cref{lem:cap-relvar-approx} that a $(1\pm \epsilon)$-approximation can be computed in $m^{O(\log n)}\eps^{-2}$ time, thereby establishing \Cref{thm:first}.

\subsection{Algorithm Description}
\paragraph{Overview.}
The algorithm is recursive. 
\eat{
There are three base cases:
\begin{enumerate}
    \item $G$ is disconnected.
    \item $p$ is larger than $n^{-10/\lambda}$.
    \item The number of vertices $n$ is a constant.
\end{enumerate}
The first base case is trivial.
In the second base case, we run a Monte Carlo sampling.
In the third base case, we run a brute force enumeration; the details are presented later.
(In the analysis for the recursive algorithm, we will assume that $n$ is at least a sufficiently large constant; that is valid because our base case applies to any constant.)
}
Before describing the algorithm formally, we give some intuition for the recursive step.
The recursive case is divided into two sub-cases depending on the maximum rank of the hyperedges. We call a hypergraph \emph{universally small} if all edge ranks are at most $n/2$; otherwise, it is said to be \emph{existentially large}.
If the hypergraph is universally small, the algorithm runs a single recursive step of random hyperedge contraction, and recursively estimates the unreliability of the contracted hypergraph. This is repeated $\poly(n)$ times to reduce the variance of the estimator, and the average of all estimates is taken as output.
If the hypergraph is existentially large, the algorithm lists all large hyperedges of rank greater than $n/2$, enumerates the first large hyperedge in the list that does not fail, and recursively estimates the unreliability of the resulting subgraph. The algorithm also handles the case that all large hyperedges fail by recursing on the (universally small) sub-hypergraph formed by deleting all large edges.

Now, we describe the algorithm formally. 

\begin{algorithm}
\caption{Unreliability$(G=(V, E), p)$}
\label{alg:enumeration}
    $n\gets |V|, \lambda \gets$ minimum cut value of $G$ computed by \Cref{thm:mincut-CQ}.\\
    \If{$G$ is disconnected}{
    \Return 1.
    }
    \If{$p^\lambda \ge n^{-10}$}{
    Run the Monte Carlo sampling algorithm.
    }
    \If{$n\le O(1)$}{
        Merge parallel hyperedges into a weighted hyperedge.\\
        Enumerate all $2^m$ possible outcomes of random hyperedge removal, where $m\le 2^n$ is the number of weighted hyperedges.\\
        \Return the total probability that the hypergraph is disconnected after random hyperedge removal, where a hyperedge of weight $w(e)$ has failure probability $p^{w(e)}$.
    }
    \eIf{$G$ is universally small}{
    $q\gets n^{-10/\lambda}$\\
    \For{$i=1$ to $2n^{12}$}{
    Sample $H_i\sim G(q)$.\\
    Recursively call Unreliability$(H_i, p/q)$ to get estimator $X_i$.
    }
    \Return the average of all $X_i$'s.
    } {
    List large hyperedges $e_1, e_2,\ldots, e_\ell$.
    Let $E_i$ be the set of first $i$ hyperedges in the list.\\
    \For{$i=0$ to $\ell-1$}{
    Recursively call Unreliability$(H_i=(G-E_i)/e_{i+1}, p)$ to get estimator $X_i$.
    }
    Recursively call Unreliability$(H_\ell = G-E_\ell, p)$ to get estimator $X_{\ell}$.\\
    \Return $p^{\ell}\cdot X_\ell + \sum_{i=1}^{\ell} p^i(1-p)\cdot X_i$.
    }
\end{algorithm}

\paragraph{Base cases.}
There are three base cases:
\begin{enumerate}
    \item $G$ is disconnected. In this case, we output 1.
    \item $p$ is larger than $n^{-10/\lambda}$. In this case, we use Monte Carlo sampling (\Cref{lem:mc-relvar})
    and take average of $n^{10}$ samples.
    \item The number of vertices $n$ is a constant. In this case, we merge all parallel hyperedges to form weighted hyperedges. We need to estimate $\ugp$ when each weighted hyperedge $e$ is removed with probability $p^{w(e)}$, where $w(e)$ is the weight of $e$. We enumerate over all possible subsets of weighted hyperedges that are deleted, and compute $\ugp$ exactly. The first step takes $O(m)$ time; the rest is $O(1)$ time. We have established the following lemma:
\begin{lemma}\label{lem:brute-force}
    When $n=O(1)$, $\ugp$ can be exactly computed in $O(m)$ time.
\end{lemma}
\end{enumerate}

\paragraph{Recursive case.}
We start by classifying hypergraphs as follows: 
\begin{define}[universally small, existentially large hypergraphs]
A hypergraph is \emph{universally small} if all hyperedges are of rank at most $n/2$.
A hypergraph is \emph{existentially large} if there exists a hyperedge of rank greater than $n/2$.
\end{define}

\paragraph{Recursive algorithm for universally small hypergraphs.}
The algorithm repeats a random contraction step independently $2n^{12}$ times.
In the $i$-th random contraction step, the algorithm samples $H_i\sim G(q)$ by contracting each edge with probability $1-q$ independently, where $q = n^{-10/\lambda}$. 
Note that $q\ge p$, otherwise we are in a base case.
Then, the algorithm recursively estimates $u_{H_i}(p/q)$. We will show later that $u_{H_i}(p/q)$ is an unbiased estimator of $\ugp$ with bounded relative variance.
After all $2n^{12}$ recursive calls, the algorithm takes the average of the estimators returned by these recursive calls to be the output.

\paragraph{Recursive algorithm for existentially large hypergraphs.}
Suppose there are $\ell$ large hyperedges, ordered arbitrarily as  $e_1, e_2, \ldots, e_\ell$.
Let $E_i$ be the subset of first $i$ hyperedges in the list; in particular, $E_0=\emptyset$.
We divide the event of hypergraph disconnection into $\ell+1$ disjoint events by enumerating the first hyperedge in the list that does not fail. Formally, for $i=0, 1, \ldots, \ell-1$, let $A_i$ be the event that first $i$ hyperedges in the list all fail, but the $(i+1)$-th hyperedge survives; Let $A_{\ell}$ be the event that all $\ell$ hyperedges fail. Then $\Pr[A_i]=p^i(1-p)$ for $i\le \ell-1$ and $\Pr[A_\ell]=p^\ell$.
Conditioned on each event $A_i$, we can remove the failed hyperedges in $E_i$ and contract the first surviving hyperedge $e_{i+1}$ to form a subgraph $H_i$.
Formally, let $H_i=(G-E_i)/e_{i+1}$ for $i=0, 1, \ldots, \ell-1$, and $H_\ell = G-E_\ell$.
The event that $G$ disconnects conditioned on $A_i$ is equivalent to $H_i$ disconnecting when each hyperedge is removed with probability $p$ independently. We have
\begin{equation}\label{eq:enumerate}
    \ugp = \sum_{i=0}^{\ell}\Pr[A_i]\cdot \Pr[G\text{ disconnects}|A_i]
    =p^\ell \cdot u_{H_\ell}(p) + \sum_{i=0}^{\ell-1} p^i(1-p) \cdot u_{H_i}(p)
\end{equation}
The algorithm runs $\ell+1=O(m)$ recursive calls on each $H_i$ to get unbiased estimators $X_i$ of $u_{H_i}(p)$. The overall estimator $X$ of $\ugp$ is then given by 
$X=p^\ell\cdot X_\ell + \sum_{i=0}^{\ell-1} p^i(1-p)\cdot X_i$.
Equation (\ref{eq:enumerate}) shows that $X$ is an unbiased estimator of $\ugp$.

The subproblems are easier because of the following reason: in $H_i$ for $i\le \ell-1$, we contracted a large hyperedge from $G$, so the number of vertices decreases by at least a half; In $H_\ell$, we removed all large hyperedges from $G$, so $H_\ell$ is universally small.


\subsection{Correctness}
In this section, we prove the following lemma that establishes correctness of the algorithm.

\begin{lemma}\label{lem:enumeration-correctness}
    \Cref{alg:enumeration} outputs an unbiased estimator with relative variance at most $1$.
\end{lemma}

Note that the base cases of disconnected $G$ and constant size output exact value of $\ugp$, and the base case of Monte Carlo sampling outputs an unbiased estimator of $\ugp$.
Also, an enumeration step in the existentially large case does not introduce variance. So, we only need to bound the relative variance introduced in the universally small case.
To do so, we first analyze the variance introduced in a random contraction step.

The key to bounding relative variance of a random contraction step is the following property of a random subgraph which we will prove later.

\begin{lemma}\label{lem:ugp-bound}
    $p^\lambda \le \ugp \le n^2p^\lambda$.
\end{lemma}
\Cref{lem:ugp-bound} provides an upper bound on the relative variance of random contraction:

\begin{lemma}\label{lem:rc-relvar}
Suppose $H\sim G(q)$ and $q\ge p$. Then, the relative variance of $u_H(p/q)$ is at most $n^2q^{-\lambda}-1$.
\end{lemma}
\begin{proof}
Because $H$ is constructed by contraction from $G$, its min-cut value $\lambda_H$ is at least the min-cut value $\lambda$ in $G$.
By Lemma \ref{lem:ugp-bound}, 
\begin{equation}\label{eq:uhpq-bound}
    u_H(p/q) \le |V(H)|^2 (p/q)^{\lambda_H} \le n^2 (p/q)^\lambda
\end{equation}
because $|V(H)|\le n$, $\lambda_H\ge \lambda$, and $q\ge p$.

$u_H(p/q)$ is an unbiased estimator of $\ugp$ by \Cref{lem:rc-unbias}. Next we bound its relative variance $\eta[u_H(p/q)]$.
By \Cref{fact:relvar-max}, the relative variance is upper bounded by $\frac{\max_H u_H(p/q)}{\ugp}-1$.
We have $\max_H u_H(p/q)\le n^2(p/q)^\lambda$ by Equation (\ref{eq:uhpq-bound}), and $\ugp \ge p^\lambda$ by \Cref{lem:ugp-bound}.
 Therefore,
\[\eta[u_H(p/q)] \le  \frac{n^2(p/q)^\lambda}{p^\lambda}-1 = n^2q^{-\lambda}-1.\qedhere\]
\end{proof}


We are now prepared to prove \Cref{lem:enumeration-correctness}.

\begin{proof}[Proof of \Cref{lem:enumeration-correctness}]
We will prove the lemma by induction on number of vertices $n$, number of hyperedges $m$, as well as the value of $\left\lceil m\ln \frac 1p \right\rceil$.
The base case of $n=O(1)$ is given by \Cref{lem:brute-force}.
The base case of $m=0$ is handled by the disconnected case in the algorithm.
These two base cases output the exact value of $\ugp$.
For the induction on $\left\lceil m\ln \frac 1p \right\rceil$, notice that this value is a positive integer because $p\in (0, 1)$ in all recursive calls.
The base case $\left\lceil m\ln \frac 1p \right\rceil \le \frac{10 m \ln n}{\lambda}$
implies $p^{\lambda}\ge n^{-10}$.
Hence, it is handled by Monte Carlo sampling in the algorithm, which outputs an unbiased estimator of $\ugp$ with relative variance at most $\frac{1}{u_G(p)}\le \frac{1}{p^\lambda} \le n^{10}$ by \Cref{lem:mc-relvar}.
After taking average of $n^{10}$ samples, the relative variance is reduced to at most $1$ by \Cref{fact:cap-relvar-average}.

For the inductive step, there are two cases. We first consider a random contraction step when the hypergraph is universally small.
This step generates $2n^{12}$ random subgraphs $H_i\sim G(q)$, where $q^\lambda = n^{-10}$.
\Cref{lem:rc-unbias} gives that $u_{H_i}(p/q)$ is an unbiased estimator of $\ugp$. \Cref{lem:rc-relvar} gives that $u_{H_i}(p/q)$ has relative variance at most $n^2q^{-\lambda}-1 = n^{12}-1$. 
By the inductive hypothesis, each subproblem returns an unbiased estimator $X_i$ of $u_{H_i}(p/q)$ with relative variance at most $1$.
By \Cref{lem:cap-relvar-concatenate}, $X_i$ is an unbiased estimator of $\ugp$ with relative variance at most $n^{12}(1+1)-1 \le 2n^{12}$. Taking average over all $2n^{12}$ estimators, $X_i$ gives an unbiased estimator of $\ugp$ with relative variance at most $1$ by \Cref{fact:cap-relvar-average}.

Next, we consider a large  hyperedge enumeration step when the input hypergraph is existentially large.
The algorithm computes an estimator $X = p^{\ell}\cdot X_\ell + \sum_{i=1}^{\ell} p^i(1-p)\cdot X_i$, where $X_i$'s are returned by recursive calls on $H_i$. By the inductive hypothesis, the recursive calls give unbiased estimators, i.e.\ $\E[X_i]=u_{H_i}(p)$. We have
\[\E[X]=p^{\ell}\cdot \E[X_\ell] + \sum_{i=1}^{\ell} p^i(1-p)\cdot \E[X_i]
=  p^\ell \cdot u_{H_\ell}(p) + \sum_{i=0}^{\ell-1} p^i(1-p) \cdot u_{H_i}(p)
= \ugp\]
where the last step is by Equation (\ref{eq:enumerate}).

$X$ is a convex combination of independent recursive estimators $X_i$, which have relative variance at most $1$ by the inductive hypothesis. Hence, $X$ also has relative variance at most $1$ by \Cref{fact:relvar-linear-combination}.

Finally, we argue that the induction is valid, i.e., that we always make progress on one of the inductive parameters in every recursive call. Whenever a  hyperedge is contracted or deleted, we decrease $n$ or $m$. It is possible that a random contraction step does not change the hypergraph. In that case, notice that $p$ is changed to $p/q$ in the subproblem. So, $m\ln \frac{1}{p}$ decreases by $m\ln \frac{1}{q} = \frac{10m\ln n}{\lambda} \ge 10$. Therefore, we also decrease $\left\lceil m\ln \frac{1}{p}\right\rceil$, and the induction is valid.
\end{proof}

In the rest of this subsection, we give a proof of our main technical lemma, \Cref{lem:ugp-bound}.

\paragraph{Proof of \Cref{lem:ugp-bound}}

The lower bound of $p^\lambda$ in \Cref{lem:ugp-bound} holds because a minimum cut fails with probability $p^\lambda$. The rest of the proof is devoted to the upper bound of $n^2p^\lambda$.

We first assume that in the hypergraph $G$, the hyperedges are partitioned into pairs of parallel  hyperedges.
This is w.l.o.g.\ because we can replace each hyperedge by two copies and change the failure probability $p$ to $\sqrt{p}$. 

We introduce some definitions that are used only in the analysis (i.e., the algorithm does not need to compute them).
For any contracted hypergraph of $G$, we choose an orientation of the hyperedges in the sense that in each hyperedge, one vertex is designated the head and all other vertices are tails.
We require the orientation to satisfy the property that any pair of parallel hyperedges (in the partition of hyperedges into pairs) have different heads.
This is always possible because the rank of each hyperedge is at least 2.
Besides this property, the choice of heads are arbitrary.
The orientation is chosen in a consistent way. That is, any fixed contracted hypergraph always chooses the same orientation throughout the analysis.

The orientation is used to define representatives of contracted supervertices.
Each contracted supervertex during the contraction process will be assigned a representative vertex, which is an original vertex contracted into the supervertex.
Initially, each vertex is its own representative. 
Whenever a hyperedge $e$ is contracted, we assign the representative of the head of $e$ to be the representative of the new contracted supervertex.

\begin{claim}\label{claim:critical-edge-construction}
    For any pair of supervertices $u \ne v$ in a contracted hypergraph of $G$, there are at least $\lambda$ hyperedges that contain at least one of $u$ or $v$ as a tail.
\end{claim}
\begin{proof}
Consider any pair of supervertices $u, v$. The (undirected) degree cut of $u$ (denoted $\partial u$) and degree cut of $v$ (denoted $\partial v$) are cuts in the original hypergraphs, and have cut value at least $\lambda$.
The hyperedges in the degree cuts are pairs of parallel hyperedges.
For a pair in $\partial u\setminus \partial v$, at least one copy contains $u$ as a tail. Similarly, for a pair in $\partial v\setminus \partial u$, at least one copy contains $v$ as a tail. Finally, for a hyperedge in $\partial u\cap \partial v$,
it contains both $u$ and $v$. The hyperedge only has one head, so either $u$ or $v$ is a tail. 
Therefore, there are at least 
\[\frac 12 (|\partial u\setminus \partial v| + |\partial v\setminus \partial u|) + |\partial u\cap \partial v| = \frac 12 (|\partial u|+|\partial v|) \ge \lambda\] hyperedges that contain $u$ or $v$ as a tail.
\end{proof}

We now proceed to prove the upper bound.

\paragraph{Exponential contraction process.}
For the sake of analysis, consider the following continuous time random process called the exponential contraction process. Let each hyperedge $e$ independently arrive at a time $Y_e$ following the exponential distribution of rate $1$. 
Then, the probability that a hyperedge does not arrive before time $\ln\frac{1}{q}$ is $e^{-\ln\frac{1}{q}}=q$. Therefore, contracting hyperedges that arrive before time $\ln\frac{1}{q}$ produces the same distribution as $H\sim G(q)$.

In the contraction process, if the hypergraph is not contracted into a single supervertex at time $\ln \frac{1}{p}$ (which happens with probability $u_G(p)$), there are at least two supervertices. Consequently, at least two vertices survive as representatives. We will show that the probability that any pair of vertices $s, t$ both survive is at most $p^\lambda$. By union bounding over all ${n\choose 2}\le n^2$ pairs of vertices, we have $u_G(p) \le n^2p^\lambda$ which completes the proof.

To bound the probability that any pair of vertices $s, t$ both survive, 
we choose a set of \emph{critical edges} during the contraction process as follows. When $s$ and $t$ are in different supervertices $\ts$ and $\widetilde{t}$, we choose $\lambda$ hyperedges that contain at least one of $\ts$ or $\widetilde{t}$ as a tail guaranteed by \Cref{claim:critical-edge-construction}; otherwise, we choose $\lambda$ arbitrary edges.
(The critical edges may change after each contraction.)
Note that whenever a critical edge arrives, one of $s$ or $t$ is no longer a representative.
Hence, if $s, t$ both survive as representatives after the contraction process, then no critical edge arrives during the contraction process.
By \Cref{lem:critical-edge-not-arrive-prob}, this happens with probability at most $e^{-\lambda \ln \frac{1}{p}} = p^\lambda$.

\begin{lemma}\label{lem:critical-edge-not-arrive-prob}
    Suppose during the exponential contraction process from time 0 to $T$, we maintain a subset of uncontracted hyperedges called the \emph{critical edges}. The critical edges could change immediately after each arrival of an uncontracted hyperedge, but do not change between two arrivals. Suppose that there are always at least $\lambda$ critical edges. 
    Then, the probability that no critical edge arrives up to time $T$ is at most $e^{-\lambda T}$.
\end{lemma}
\begin{proof}
    The proof is by induction on the number of uncontracted hyperedges $m$. The base case is $m=\lambda$, where the set of critical edges cannot change, and the probability that no critical edge arrives is $e^{-\lambda T}$.

    For the inductive case, let $m_{\mathrm{crit}}\ge\lambda$ be the current number of critical edges. We bound the probability that no critical edge arrives up to time $T$ (denoted $p_G(T)$) by integrating over the earliest time $t$ that a hyperedge $e_i$ arrives. By \Cref{fact:exp-min}, $t$ follows the exponential distribution of rate $m$, and $e_i$ is not a critical edge with probability $1-\frac{m_{\mathrm{crit}}}{m}$.
    After the earliest hyperedge $e_i$ arrives, we apply the inductive hypothesis on hypergraph $G/e_i$ and remaining time $T-t$.
    \begin{align*}
        p_G(T) &= \int_0^T me^{-mt}\diff t \cdot \left(1-\frac{m_{\mathrm{crit}}}{m}\right) \cdot p_{G/e_i}(T-t) + \int_T^{\infty} me^{-mt}\diff t \cdot 1\\
        &\le e^{-mT} + \int_0^T me^{-mt}\diff t \cdot \left(1-\frac{\lambda}{m}\right)\cdot e^{-\lambda(T-t)}\\
        &= e^{-mT} + e^{-\lambda T}\int_0^T (m-\lambda) e^{-(m-\lambda) t}\diff t\\
        &= e^{-mT} + e^{-\lambda T} (1-e^{-(m-\lambda) T}) = e^{-\lambda T},
    \end{align*}
    which completes the inductive case.
\end{proof}

\subsection{Running Time}\label{sec:alg1-runtime}
In this section, we prove the following lemma of running time of the algorithm.
\begin{lemma}\label{lem:alg1-runtime}
    The expected running time of \Cref{alg:enumeration} is $m^{O(\log n)}$.
\end{lemma}

\paragraph{Size decrease bound.}
In order to bound the size of the recursion tree, we ideally want each random contraction step to reduce the number of supervertices by a constant factor, so that the recursion tree has depth $O(\log n)$. 
This does not generally hold with high probability when some hyperedge's rank is large.
As an extreme example, consider the hypergraph that consists of $n$ distinct hyperedges of rank $n-1$. The probability to contract nothing is $q^n$, while the contraction step is repeated $n^2q^{-\lambda} = n^2 q^{-(n-1)}$ times, so we expect to see at least one bad subproblem where the number of vertices does not decrease. If this happens, then the size of the recursion tree would be unbounded.

This is where the maximum rank assumption in the universally small case becomes useful: we can control the probability to get a bad subproblem to be small enough, which is crucial in bounding the size of the recursion tree.

\begin{lemma}\label{lem:rc-small-rank-size-bound}
Fix constants $A > B > 1$. 
Suppose that all hyperedges have rank at most $R=n/A$, and $n$ is larger than some constant depending on $A$.
Let $n^*=\lceil BR\rceil$, and let $H\sim G(q)$.
Then, $\Pr[|V(H)|\ge n^*+1]<n^{A^2/B}q^{B\lambda }$. 
\end{lemma}

\begin{proof}
    Define $t(G)$ to be the stopping time in the contraction process when the vertex size of the contracted hypergraph decreases to at most $n^*$.
    By definition
    \[\Pr[|V(H)|>n^*] = \Pr\left[t(G)>\ln \frac{1}{q}\right].\]
    By Markov's inequality,
    \[
        \Pr\left[t(G)>\ln \frac{1}{q}\right] = \Pr[e^{B\lambda  \cdot t(G)}>q^{-B\lambda}]
        \le \frac{\E[e^{B\lambda \cdot t(G)}]}{q^{-B\lambda}} = \E[e^{B\lambda \cdot t(G)}]q^{B\lambda}.\\
    \]
    By choosing the same parameters $A, B$ in \Cref{lem:size-decrease-induct} (which we will prove later), 
    \[\Pr\left[t(G)>\ln \frac{1}{q}\right] \le \E[e^{B\lambda \cdot t(G)}]q^{B\lambda}
        \le \left(\frac{en}{n^*}\right)^{A\ln n} q^{B\lambda} 
        \le (Ae/B)^{A\ln n}q^{B\lambda}
        \le n^{A(1+\ln\frac{A}{B})}q^{B\lambda}
        \le n^{A^2/B} q^{B\lambda},
    \]
    where the last inequality uses that $1+\ln x\le x$ for all $x\ge0$.
\end{proof}

\begin{corollary}\label{cor:rc-0.5n-bound}
    Assume that all hyperedges have rank at most $n/2$ and $n$ is larger than a suitable constant. Let $H\sim G(q)$. Then, $\Pr[|V(H)|\ge 0.8n] \le n^{2.7}q^{1.5\lambda}$.
\end{corollary}
\begin{proof}
    In \Cref{lem:rc-small-rank-size-bound}, set $A=2$ and $B=1.5$.
    We have $n^*+1 = \lceil 1.5n/2 \rceil+1 \le 0.8n$. Then, $\Pr[|V(H)|\ge 0.8n]<n^{8/3}q^{1.5\lambda} \le n^{2.7}q^{1.5\lambda}$.
\end{proof}

\begin{lemma}\label{lem:size-decrease-induct}
    Fix constants $A > B > 1$.
    Suppose the maximum rank is upper bounded by $R$ and $R$ is larger than some constant depending on $A$. 
    For any hypergraph $\tG$ formed by contraction from $G$, define $t(\tG)$ to be the following stopping time: In a contraction process starting at $\tG$, the vertex size of the contracted hypergraph decreases to at most $n^* = \lceil BR \rceil$.
    Suppose $|V(\tG)|\le N$, where $N=AR$.
    Let $\tlambda$ be the min-cut value in $\tG$. 
    Then,
    $$\E\left[e^{B \tlambda \cdot t(\tG)}\right]\le \max\left\{ \left(\frac{|V(\tG)|}{n^*/e}\right)^{A \ln N}, 1\right\}.$$
\end{lemma}
\begin{proof}
    Our proof is by induction on $\tn = |V(\tG)|$.  As the base case, when $ \tn \le n^*$, by definition $t(\tG)=0$ and $e^{B\tlambda  \cdot t(\tG)}=1$, so the statement holds.

    Next consider the inductive step where $\tn > n^*$. Let $\bar{r}=\frac{1}{\tm}\sum_{e\in E(\tG)} r(e)$ be the average rank in $\tG$, where $\tm = |E(\tG)|$.

    Let $\tau=\min_{e\in E(\tG)} X_e$ be the earliest arrival time of a hyperedge in $\tG$. Because the $X_e$'s are sampled from i.i.d.\ exponential distributions of rate 1, the random variable $\tau$ follows the exponential distribution of rate $\tm$ by \Cref{fact:exp-min}. Note that any degree cut in $\tG$ has value at least the min-cut value $\tlambda$ in $\tG$.
    By summing over degrees of all vertices, we have
    \[\tn \tlambda \le \sum_{e\in E(\tG)} r(e) = \tm \bar{r}
    \implies \tm \ge \frac{\tn \tlambda}{\bar{r}} .\]
    Denote $\rate(\cdot)$ to be the rate of a exponential distributed random variable.
    Because $\tn > n^*$ and $\bar{r}\le R$, we have $\text{rate}(\tau) = \tm \ge \frac{\tn \tlambda}{\bar{r}} > \frac{n^* \tlambda}{R} \ge B\tlambda$.

    By the moment generating function of the exponential distribution in \Cref{fact:exp-mgf},
    
    \begin{equation}\label{eq:NBtau}
    \E\left[e^{B \tlambda \tau}\right]
    =\frac{1}{1-\frac{B\tlambda}{\text{rate}(\tau)}}
    \le \frac{1}{1-\frac{B\tlambda}{\tn \tlambda / \bar{r}} }
    = \frac{1}{1-\frac{B\bar{r}}{\tn}}
    \end{equation}

    Now, consider the distribution of $t(\tG)$ after revealing $\tau$. To decrease size from $\tn > n^*$ to below the threshold $n^*$, at least one  hyperedge needs to be contracted; so $t(\tG)\ge \tau$.
    Moreover, by the memoryless property of the exponential distribution in \Cref{fact:exp-memoryless}, $\tau$ and $t(\tG)-\tau$ are independent.
    If we further reveal the information that the  hyperedge that arrives the earliest is $e_i$, then we know the random process after the earliest arrival is equivalent to running the contraction process starting from $\tG/e_i$. Thus, $t(\tG)-\tau$ follows the same distribution as $t(\tG/e_i)$ conditioned on $e_i$ arrives first. Because the  hyperedges follow i.i.d.\ exponential distribution, they are equally likely to arrive the earliest. Therefore, unconditionally we have that $t(\tG)-\tau$ follows distribution of $t(\tG/e_i)$ for each $e_i$ with probability $\frac{1}{\tm}$.
    Formally,
    \begin{equation}\label{eq:NBt-induction}
        \E\left[e^{B\tlambda  \cdot t(\tG)}\right] = \E\left[e^{B\tlambda  \tau}\right]\cdot \E\left[e^{B\tlambda (t(\tG)-\tau)}\right]
        = \E\left[e^{B\tlambda  \tau}\right] \sum_{e_i\in E(\tG)} \frac{1}{\tm} \E\left[e^{B\tlambda \cdot t(\tG/e_i)}\right]
    \end{equation}
    We now apply the inductive hypothesis to bound the term $\E[e^{B\tlambda \cdot t(\tG/e_i)}]$ for each $e_i$ separately.
    Note that the min-cut value $\lambda_i$ in $\tG/e_i$ is at least $\tlambda$ because a contraction cannot decrease the min-cut value.
    \begin{align*}
        \E\left[e^{B\tlambda \cdot t(\tG/e_i)}\right]
        &\le \E\left[e^{B\lambda_i \cdot t(\tG/e_i)}\right] 
        \le \max\left\{\left(\frac{\tn-r(e_i)+1}{n^*/e}\right)^{A\ln N}, 1\right\}\\
        &= \frac{(\max\{\tn-r(e_i)+1, n^*/e\})^{A\ln N}}{(n^*/e)^{A\ln N}}
        \le \frac{(\max\{\tn-r(e_i)+1, \tn/e\})^{A\ln N}}{\tn^{A\ln N}}\cdot \left(\frac{e\tn}{n^*}\right)^{A\ln N}\\
        &= \max\left\{\left(1-\frac{r(e_i)-1}{\tn}\right)^{A\ln N}, N^{-A}\right\}\cdot \left(\frac{e\tn}{n^*}\right)^{A\ln N}.
    \end{align*}

    Next, we prove that $\max\{(1-\frac{r-1}{\tn})^{A\ln N}, N^{-A}\} \le 1-\frac{B r}{\tn}$ for all $r\in [2, R]$.
    For the second term, we have $N^{-A} < \frac 1N $, while
    \[1-\frac{B r}{\tn} \ge 1-\frac{BR}{\tn} \ge 1-\frac{n^*}{n^* + 1}  = \frac{1}{n^*+1} \ge \frac{1}{N}.\]
    
    For the first term, define $x=\frac{r(e_i)-1}{\tn}$ in terms of $r(e_i)$, which satisfies $x\in [\frac{1}{\tn}, \frac{R-1}{\tn}]$. Then, $r(e_i) = \tn x+1$.
    We want to bound $(1-\frac{r-1}{\tn})^{A\ln N}\le 1-\frac{Br}{\tn}$ for all $r\in [2, R]$, which is equivalent to 
    \begin{equation}\label{eq:x-condition}
    (1-x)^{A\ln N} \le 1-Bx -\frac{B}{\tn}.
    \end{equation}
    Note that the LHS is convex on $(0,1)$ when $A\ln N\ge 2$, and the RHS is linear. So we only need to prove (\ref{eq:x-condition}) for the two endpoints $x=\frac{1}{\tn}$ and $x=\frac{R-1}{\tn}$. For $x=\frac{1}{\tn}$, we have
    \[\left(1-\frac{1}{\tn}\right)^{A \ln N} 
    \le e^{-A \ln N / \tn} 
    \le 1-\frac{A\ln N}{2\tn} 
    \le 1-\frac{2B}{\tn}\]
    Here, the second inequality is by $e^{-x}\le 1-\frac{x}{2}$ for $x\in [0,1]$ or $A\ln N/\tn \le 1$, which holds when $N$ is larger than some constant (depending on $A$). The last inequality holds when $\ln N \ge 4$.

    Next, we consider the second endpoint $x=\frac{R-1}{\tn}$.
    The assumption $N=AR \ge \tn$ implies $\frac{R}{\tn}\ge \frac{1}{A}$. So, we have
    \[\left(1- \frac{R-1}{\tn}\right)^{A\ln N}\le  \left(1- \frac{1}{A} + \frac{1}{\tn}\right)^{A\ln N}  \le \frac{1}{N}\]
    Here, the last inequality is equivalent to $1-\frac{1}{A}+\frac{1}{\tn} \le e^{-1/A}$, which holds when $\tn > 3A^2$ and $A>1$.

    Now, we have
    \[ \E\left[e^{B\tlambda \cdot t(\tG)}\right]
    \stackrel{(\ref{eq:NBt-induction})}= \E\left[e^{B\tlambda \tau}\right] \cdot \frac{1}{\tm} \sum_{e_i\in E(\tG)} \E\left[e^{B\tlambda \cdot t(\tG/e)}\right]
    \stackrel{(\ref{eq:NBtau})}\le\frac{1}{1-\frac{B\bar{r}}{\tn}} \cdot \frac{1}{\tm} \sum_{e_i\in E(\tG)} \E\left[e^{B\tlambda \cdot t(\tG/e)}\right]
    \]
    \[
    \le \frac{1}{1-\frac{B\bar{r}}{\tn}} \cdot  \frac{1}{\tm} \sum_{e_i} \left(1-\frac{B r(e_i)}{\tn}\right)\cdot \left(\frac{e\tn}{n^*}\right)^{A\ln N}  \le \left(\frac{e\tn}{n^*}\right)^{A\ln N} \]
    which establishes the inductive case.
\end{proof}

\paragraph{Proof of \Cref{lem:alg1-runtime}}

We color each recursive call as black or red. 
Intuitively, they represent a ``success'' or ``failure'' in recursive calls respectively.
A call is black if it decreases size by a constant factor, more precisely when $|V(H)|\le 0.8n$.
Otherwise, the call is red.

The recursion tree has the following properties:

\begin{enumerate}
    \item Each subproblem makes at most $M=m^{O(1)}$ recursive calls. This is guaranteed by the algorithm.
    \item The algorithm reaches base case after $O(\log n)$ black recursive calls. This is because each black recursive call decreases $n$ by a constant factor, and $O(\log n)$ black calls reduces $n$ to a constant, which is a base case.
    \item For each subproblem, the expected number of red recursive calls is at most $1/2$.
    This is because in random contraction, we have $O(n^2q^{-\lambda})$ recursive calls, and each call fails (to get size decrease) with probability at most $n^{2.7}q^{1.5\lambda}$ by \Cref{cor:rc-0.5n-bound}. The expected number of red calls is their product, which is upper bounded by $n^{4.7}q^{0.5\lambda}=o(1)$ when $q^\lambda = n^{-10}$.
\end{enumerate}

\Cref{lem:recursive-tree-size-bound} below shows that these properties give a upper bound of $m^{O(\log n)}$ on the number of recursive calls. 
If we charge the time of sampling a random contracted hypergraph to the subproblem on the contracted hypergraph, then each subproblem spends $O(nm)$ time outside the recursive calls. Therefore, the overall expected running time is $m^{O(\log n)}$. This concludes the proof of \Cref{lem:alg1-runtime}.


\begin{figure}
    \centering
  \includegraphics[width=.2\textwidth]{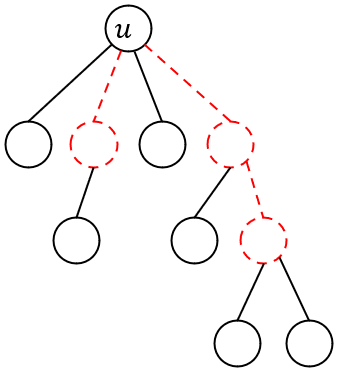}
  \caption{A depiction of a portion of the computation tree. The failed recursive calls are shown in dashed red, while the successful ones are shown in solid black. \Cref{lem:recursive-tree-size-bound} analyzes the expected size of the recursion tree.}
  \label{fig1}
\end{figure}
\begin{lemma}\label{lem:recursive-tree-size-bound}
    Suppose in a randomly growing tree, each node $u$ is either a leaf, or has $M(u)\le M$ children, where $M=\Omega(n)$ is a parameter. Each edge from $u$ to its children is colored red with probability $f(u)$ such that $M(u)f(u) \le \theta = \frac 12$, and black otherwise.
    The different children at a parent node are independent (including independence betwen the parent-child edges); Also, a subtree is independent of everything outside the subtree.
    Moreover, on any path from root to leaf, there can be at most $L$ black edges. Then, the expected number of nodes in the tree is at most $M^{O(L)}$.
\end{lemma}
\begin{proof}
    We say a node $w$ is a red descendant of node $u$ if $u$ is an ancestor of $w$, and the path from $u$ to $w$ is formed by red edges only. See \Cref{fig1} for an illustration.

    Let $K$ be the number of red descendants of some node $u$.
    Let $K_i$ be the number of red descendants of $u$ that are $i$ steps deeper than $u$, so that $K=\sum_{i\ge 1}K_i$. We have $\E[K_1] = M(u)f(u) \le \theta$.
    Inductively, suppose at level $i$, there are $K_i$ red descendants $\{u_1, \ldots, u_{k_i}\}$. By definition, the red descendants in level $i+1$ must be children of red descendants in level $i$.
    Each $u_j$ will generate at most $\theta$ red edges in expectation. So, $\E[K_{i+1}] \le \theta\cdot \E[K_i]$.
    By induction, $\E[K_i] \le \theta^i$.
    Note that the sum $\sum_{i\ge 1}\theta^i = \frac{\theta}{1-\theta} \le 1$ converges and $K_i$'s are nonnegative, so we can apply Fubini's theorem to get
    \[\E[K] = \E\left[\sum_{i\ge 1}K_i\right] = \sum_{i\ge 1}\E[K_i] \le \sum_{i\ge 1}\theta^i \le 1.\]

    Let $V_i$ be the set of nodes that have $k$ black edges from the root.
    We prove by induction that $\E[|V_k|]\le 2(2M)^k$. It follows that the expected number of nodes is $\E[\sum_{k=0}^{L}|V_k|] = O((2M)^L) = M^{O(L)}$.
    The base case for $k=0$ is the expected number of red descendants of the root, as well as the root itself, which is at most $1+1=2=2(2M)^0$.
    Next consider the inductive step. For any node $w\in V_{k+1}$, let $(u, v)$ be the black edge closest to $w$ on the path from root to $w$. Then, $u\in V_k$, and $w$ is either a red descendent of $v$ or $v$ itself. Each node $u$ in $V_k$ generates $M(u)$ children $v$, and each $v$ generates at most $1$ red descendants in expectation. Therefore, $\E[|V_{k+1}|] \le \E[|V_k|]\cdot M\cdot (1+1) \le 2(2M^{k+1})$.
\end{proof}

\section{Random Contraction with DNF Sampling}
\label{sec:dnf-sampling}
In this section, we strengthen the previous algorithm to obtain a running time of $m \cdot n^{O(\log n\cdot \log \log \frac{1}{\delta})}$, at the cost of an additive error of $\delta$, i.e.\ the output estimator is within $(1\pm \eps)\ugp\pm \delta$ whp (\Cref{thm:second}).
When $\delta=2^{-\poly(n)}$, the running time of the algorithm is $m \cdot n^{O(\log^2 n)}$.
We show that the algorithm outputs an estimator of $\ugp$ with bias at most $\delta$ and $\delta$-capped relative variance $O(1)$. \Cref{thm:second} then follows by \Cref{lem:cap-relvar-approx}.

We denote $N=\log_2 \frac{1}{\delta}$.
We can assume wlog $N\ge \log_2 n$, i.e.\ $\delta \le \frac{1}{n}$. This is because we can run a simple Monte Carlo simulation when $\delta\ge \frac{1}{n}$ in $O\left(\frac{\log n}{\delta\eps^2}\cdot nm\right) = \tO(n^2 m\eps^{-2})$ time by \Cref{lem:mc-algo}.

\subsection{Algorithm Description}
The algorithm is recursive. We start by defining the simple base cases. Then, we introduce the definition of large hyperedges, which characterize the last base case.
Finally, we define the recursive cases.

\subsubsection{Base Cases}
There are four base cases, three of which are the following:
\begin{enumerate}
    \item When the number of vertices $n$ is a constant, we enumerate all possible outcomes by brute force, which is identical to \Cref{lem:brute-force}.
    \item When the hypergraph is already disconnected, output 1.
    \item When $p^\lambda < 2^{-3N}$, output 0.
\end{enumerate}

These three base cases are deterministic.
The algorithms for the first two base cases return the exact value of $\ugp$.
The third case has an additive bias of $\ugp$, which is at most $n^2p^\lambda$ by \Cref{lem:ugp-bound}.
We assumed $N\ge \log_2 n$ and $p^\lambda < 2^{-3N}$, so $n^2p^\lambda \le (2^N)^2\cdot 2^{-3N} \le 2^{-N}\le \delta$.

The fourth base case is called full revelation; it will be described later in the section.

\subsubsection{Large and Small Hyperedges}
Before proceeding with the formal description of the remaining algorithm, let us provide some intuition.
Recall from \Cref{alg:enumeration} that the large hyperedges are the bottleneck in the random contraction algorithm: we branch $m$ times to halve the number of vertices, which leads to the $m^{O(\log n)}$ running time. However, if we ignore the small hyperedges and consider the hypergraph with large hyperedges only, it turns out that the structure of cuts becomes much simpler. This motivates us to partition the set of hyperedges $E$ into two sets, $\Elarge$ and $\Esmall$. Intuitively, these are sets of hyperedges of large and small rank respectively, but for technical reasons, the precise definition needs to be more nuanced. 

We now formally define the set $\Elarge$.
It depends on \emph{phase nodes} in the recursive computation tree, which we define first. 
Initially, the root of the computation tree is a phase node. For any non-root node $w$ of the computation tree, let $u$ be the closest ancestor node of $w$ that is a phase node (which exists because the root is a phase node). $w$ is a phase node if and only if the number of vertices $n_u$ in $u$ and $n_w$ in $w$ satisfy $n_w \le 0.8n_u$.
If $w$ is not a phase node, such a $u$ is called the \emph{phase ancestor} of $w$.
Define a \emph{phase} with phase node $u$ to be all computation nodes whose phase ancestor if $u$, as well as $u$ itself.
See \Cref{fig2} for an illustration.

Given the definition of phase nodes, we define $\Elarge$ as follows:
In a phase node, $\Elarge$ is the set of all hyperedges of rank $> n/2$. In a non-phase node $w$, $\Elarge$ is inherited from its phase ancestor $u$, i.e.\ the set of all hyperedges of rank $> n_u/2$ in $u$. Let $\Glarge$ denote the hypergraph $(V(G), \Elarge)$. We let $\Esmall$ be the complement set of hyperedges $E\setminus \Elarge$, and $\Gsmall = (V(G), \Esmall)$.

\begin{figure}
    \centering
  \includegraphics[width=.35\textwidth]{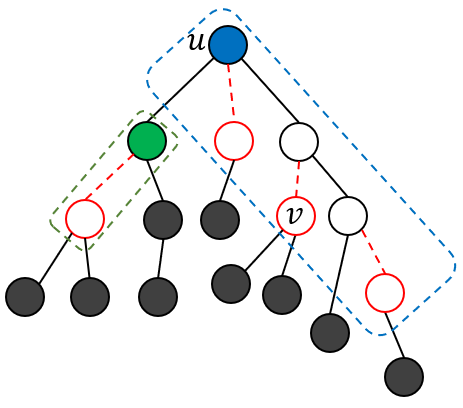}
  \caption{A depiction of phases in the computation tree. The filled in nodes are phase nodes. The blue and green nodes respectively root the blue and green phases. Each phase can contain successful recursive steps, shown by solid black edges and black nodes, and failed recursive steps, shown by dashed red edges and red nodes. In a phase, every node has a phase ancestor which is the root node of the phase; for instance, $u$ is the phase ancestor of $v$ (and of every other node in the blue phase).}
  \label{fig2}
\end{figure}

\subsubsection{The Last Base Case: Full Revelation}
We are now ready to describe the last base case that we call full revelation. Let $\beta=\lambda - \lambda_L$, where $\lambda_L$ is the min-cut value in $\Glarge$.
The last base case is invoked when $\beta<\lambda/N$.

The algorithm samples a random subgraph $H\sim G(p)$ conditioned on the event that the contracted hyperedges in $\Elarge$ do not contract the whole hypergraph into a singleton.
This is done in two steps. First, we write a DNF formula for the disconnection event in $\Glarge$, and apply \Cref{lem:dnf-sampling} to contract each hyperedge in $\Elarge$ with probability $1-p$ conditioned on the event that $\Glarge$ is not contracted into a single vertex. Second, we directly sample the remaining uncontracted hyperedges in $\Esmall$, that is we independently contract each of those hyperedges with probability $1-p$. 
The resulting hypergraph $H$ follows the desired distribution.

The algorithm repeats $8n^2$ independent samples of the above process to obtain samples $H_i$, and estimates $X_i=0$ if $H_i$ is contracted into a singleton, and $X_i=1$ otherwise. Let $X$ be the average of all these estimators $X_i$. 
Next, we use the DNF counting algorithm in \Cref{lem:dnf-unbias} to get an unbiased estimator $Z$ of $u_{\Glarge}(p)$. The product $XZ$ is the estimator of $\ugp$ output by the algorithm.

\bigskip

We are now left to describe the recursive step of the algorithm. For this purpose, we need to first establish some properties of large edges:

\subsubsection{Properties of Large Edges}


The first property is that the association of $\Elarge$ with large ranks and $\Esmall$ with small ranks is approximately correct:
\begin{fact}\label{fact:elarge-rank-range}
    Any hyperedge in $\Elarge$ has rank at least $0.3n$, and any hyperedge in $\Esmall$ has rank at most $0.7n$.
\end{fact}
\begin{proof}
    Let $w$ be the current computation node. If $w$ is a phase node, then the fact is by definition of $\Elarge$.
    Else, let $u$ be the phase ancestor of $u$.
    For any hyperedge $e\in \Elarge(w)$, 
    $$r_w(e) \ge r_u(e) - (n_u-n_w) \ge 0.5n_u - 0.2n_u = 0.3 n_u \ge 0.3 n_w.$$
    For any hyperedge $e\in \Esmall(w)$,
    \[
    r_w(e) \le r_u(e) \le 0.5n_u \le \frac{0.5n_w}{0.8} \le 0.7n_w.\qedhere
    \]
\end{proof}

The algorithm only contracts non-trivial hyperedges during recursion, i.e., hyperedges that are contracted into a singleton supervertex are removed. We assert that all large hyperedges are candidates for contraction:

\begin{fact}\label{fact:elarge-inherit}
    Suppose $w$ is a non-phase node and $\Elarge(w)$ is inherited from the phase ancestor $u$ of $w$. Then, every hyperedge in $\Elarge(u)$ still appears in $\Elarge(w)$, but may be partially contracted.
\end{fact}
\begin{proof}
 If there is a hyperedge $e$ in $\Elarge(u)\setminus\Elarge(w)$, then that hyperedge is contracted to a single vertex at node $w$. But then $n_w\le n_u-r_u(e)+1\le n_u/2$, so $w$ is a phase node by definition, a contradiction. 
\end{proof}

Finally, we come to the most important property, that of the simple structure of cuts in $\Glarge$.
To explain this, let us introduce the following property: 

\begin{define}[pairwise intersecting property]
A set of hyperedges is \emph{pairwise intersecting} if any two hyperedges in the set share at least one vertex.
\end{define}

Note that in particular, any set of hyperedges of rank $> n/2$ satisfy the pairwise intersecting property because the sum of ranks of any two hyperedges is more than $n$. We assert that this property continues to hold with our modified definition of large hyperedges:

\begin{fact}\label{fact:large-pairwise}
    The set of hyperedges $\Elarge$ satisfies the pairwise intersection property.
\end{fact}
\begin{proof}
    By definition, $\Elarge$ is  either the set of hyperedges of rank $> n/2$, or inherited from the phase ancestor $u$.
    In the former case, there cannot be two disjoint hyperedges because their ranks sum to $>n$.
    In the latter case, the property is inductively inherited from $u$. 
\end{proof}

This allows us to conclude that $\Glarge$ gets disconnected if and only if any degree cut fails.

\begin{lemma}\label{lem:all-large-degree-cuts}
    If a hypergraph $G$ satisfies the pairwise intersecting property, then $G$ disconnects if and only if some degree cut fails. In particular, by \Cref{fact:large-pairwise}, this is true for $\Glarge$.
\end{lemma}

\begin{proof}
Let $H$ be the subgraph of surviving hyperedges.
Consider $H$'s connected components.
Any non-singleton connected component must contain at least one hyperedge.
For any two non-singleton components, their hyperedges are disjoint, which contradicts the pairwise intersecting property.
Therefore, there can only be at most one non-singleton connected component in $H$.
$H$ is either an empty hypergraph, i.e.\ no hyperedges, or it has one non-singleton component with (zero or more) singletons.
This means $H$ is disconnected if and only if it contains an isolated singleton, which is further equivalent to some degree cut failing in $G$.
\end{proof}

Now, the event of $\Glarge$ getting disconnected can be written into a DNF formula $F$, where each variable represents the failure of a hyperedge and each clause represents the failure of a degree cut in $G_{large}$ as the logical AND of the failure of all hyperedges in the cut.
$F$ has $n$ clauses and $m$ variables.
Therefore, by \Cref{lem:dnf-unbias,lem:dnf-sampling}, we have the following:
\begin{lemma}\label{lem:dnf-graph}
    We can do the following in $O(n^2m)$ time: 
    \begin{enumerate}
        \item Compute an unbiased estimator of $u_{\Glarge}(p)$ with relative variance at most $1$.
        \item Sample a contracted hypergraph $H$ of $\Glarge$ where every hyperedge in $\Glarge$ is contracted with probability $1-p$, conditioned on the event that $H$ is not a singleton supervertex.
    \end{enumerate}
\end{lemma}


\subsubsection{Recursive Cases}
In universally small hypergraphs, the algorithm is identical to \Cref{alg:enumeration}. We run random contraction with $q^{\lambda} = n^{-10}$, and repeat the step $16n^{12}$ times. 

The algorithm for the existentially large case is now different because we cannot afford to enumerate $m$ events.
Recall that when there exist large hyperedges, the algorithm divides into two cases depending on the value of $\beta=\lambda-\lambda_L$, where $\lambda_L$ is the min-cut value in $\Glarge$. 
When $\beta<\lambda/N$, we get full revelation that we have already described as the last base case.

We call the remaining case when $\beta\ge \lambda/N$ \emph{partial revelation}.
In that partial revelation case, the algorithm still runs a form of random contraction, but only in a subspace of the entire probability space.
We use the parameter $\beta$ to control the speed of random contraction.
By \Cref{lem:all-large-degree-cuts}, $\lambda_{L}=\min_u d_{\rm large}(u)$, i.e., the minimum degree of a vertex in $\Glarge$. 
Intuitively, $\beta$ is used to control the number of small hyperedges in each degree cut, which measures the speed of random contraction when no large hyperedges get contracted.
Note that $0\le \beta \le \lambda$.
Ideally, we want to decrease $\beta$ to as small as $\lambda/N$, which reduces to the full revelation case.
However, $\beta$ can be non-monotone as both $\lambda$ and $\lambda_L$ can increase because of contraction during recursion.
So, we define another parameter $\gamma$ that can be related to $\beta$ to bound the depth of recursion in a phase. Let $\gamma=\ell-\lambda_L$, where $\ell=|\Elarge|$ is the number of large hyperedges. We show that unlike $\beta$, $\gamma$ is monotone in a phase:
\begin{lemma}\label{lem:gamma-monotone}
    Suppose $v, w$ are nodes in the same phase. 
    If $w$ is a descendant of $v$, then $\gamma_v \ge \gamma_w$.
\end{lemma}
\begin{proof}
    Since $\Elarge$ is the same set of hyperedges within a phase by \Cref{fact:elarge-inherit}, we have $\ell_v = \ell_w$. Now, since $w$ is a descendant of $v$, the hypergraph $\Glarge(w)$ is formed by contracting some set of hyperedges in $\Glarge(v)$. Hyperedge contractions cannot decrease the value of the minimum cut; hence, $\lambda_L(w) \ge \lambda_L(v)$. The lemma follows.
\end{proof}


\paragraph{Algorithm for partial revelation.}
The algorithm runs random contraction at a more aggressive rate $q^{\beta} = n^{-700}$.\footnote{The reason for this large polynomial in $n$ is as follows. In the proof of \Cref{lem:alg2-cap-relvar}, the algorithm needs to repeat the random contraction $O(n^4q^{-\beta})$ times, and we show in \Cref{lemma:failure-probability} that each trial has failure probability $n^2q^{1.01\beta}$. We need their product $O(n^6q^{0.01\beta})$ to be $o(1)$, hence the choice $q^\beta=n^{-700}$.} 
This is done in two steps. First, we write a DNF formula for the disconnection in $\Glarge$, and apply \Cref{lem:dnf-graph} to contract each hyperedge in $\Elarge$ with probability $1-q$ conditioned on $\Glarge$ not being contracted into a singleton. Second, we independently contract each uncontracted hyperedge in $\Esmall$ with probability $1-q$.
The resulting hypergraph $H$ follows the distribution of $H\sim G(q)$ conditioned on the event that the contracted hyperedges in $\Elarge$ do not contract the whole hypergraph into a singleton.

The algorithm repeats $32n^{704}$ independent samples $H_i$, and recursively computes a (biased) estimator $X_i$ of $u_{H_i}(p/q)$.
Let $X$ be the average of all these estimators $X_i$. 
Next, we use the DNF counting algorithm in \Cref{lem:dnf-graph} to get an unbiased estimator $Z$ of $u_{\Glarge}(q)$. The product $XZ$ is the estimator of $\ugp$ output by the algorithm. 

\subsection{Bias of the Estimator}
We first show that all base cases have bias at most $\delta$, and the recursive steps are unbiased.
Then, we prove by induction that the recursion keeps the same bound $\delta$ on bias.

We introduce some notations when $\Elarge$ and $\Esmall$ are uniquely defined in context.
Let $G(p_1, p_2)$ be the random subgraph formed by independently contracting each hyperedge in $\Elarge$ with probability $1-p_1$, and each hyperedge in $\Esmall$  with probability $1-p_2$.
Let $D_L$ be the event that in some random contraction, the contracted hyperedges in $\Elarge$ do not contract the whole hypergraph into a singleton.

\paragraph{Base cases.}
The first base case of $n=O(1)$ outputs the exact value of $\ugp$ by \Cref{lem:brute-force}.
The second base case of disconnected $G$ is trivial.
In the third base case, the bias is $0-\ugp\in[-\delta, 0]$.
Next, we prove that the algorithm in full revelation case is unbiased.

\begin{lemma}\label{lem:full-revelation-unbias}
    The algorithm in the full revelation case outputs an unbiased estimator of $\ugp$.
\end{lemma}
\begin{proof}
    The algorithm first samples $H_L\sim G(p, 0)\,|\,D_L$, then forms $H_S$ by deleting all hyperedges in $\Elarge$ from $H_L$, and finally samples $H\sim H_S(p)$.
    In other words, the subgraph $H_S$ is sampled by either contracting or deleting all large hyperedges, and then $H$ is sampled from $H_S$ by randomly contracting small hyperedges. Let $a$ be the number of contracted hyperedges in $\Elarge$.
    Then, $H_S$ is sampled with probability $p^a(1-p)^{\ell-a} / u_L$, where $u_L=u_{G_{large}}(p)$.
    
    The step estimates $\ugp/ u_L$ by $X_H$, where $X_H=0$ if $H$ is contracted into a singleton, and $X_H=1$ otherwise.
    We have $\E[X_H|H_S] = u_{H_S}(p)$ by \Cref{lem:rc-unbias}.
    
    \[ \E[X_H] = \E_{H_S}[\E[X_H|H_S]] = \E[u_{H_S}(p)]
    =\sum_{H_S}\frac{p^a(1-p)^{\ell-a}}{u_L} u_{H_S}(p) \]
    The sum is over all $H_S$ sampled by either contracting or deleting all large hyperedges, conditioned on that it is not contracted into singleton by large hyperedges. The condition can be dropped because when it doesn't hold we have $u_{H_S}(p)=0$. Then the sum is $\ugp/u_L$.

    The algorithm outputs the product $XZ$. $X$ is an average of $X_H$; its expectation is $\ugp/u_L$. $Z$ has expectation $u_L$. So the total expectation is $\ugp$.
\end{proof}


\paragraph{Recursive cases.}
A random contraction step in the universally small case is unbiased by \Cref{lem:rc-unbias}. So, we only need to show this for the partial revelation case. We do this in two steps. First, we assume that the inductive subproblems in this case return exact estimators, and show that the resulting estimator after this step is unbiased. Then, we use this fact to show that if the inductive subproblems return biased estimators, then the bias does not increase after the partial revelation step.

Define a \emph{partial revelation step} to be that of the algorithm  in the partial revelation case, except that we now directly use $u_{\Glarge}(q)$ times average of $u_{H_i}(p/q)$ as the estimator instead of recursively estimating them.

\begin{lemma}\label{lem:partial-revelation-bias}
     A partial revelation step is an unbiased estimator of $\ugp$.
\end{lemma}
\begin{proof}
    A partial revelation step can be expressed as $H_L\sim G(q, 0)\,|\,D_L$, and $H\sim H_L(0, q)$.
    Here, the sampling is viewed as two steps: First, the algorithm samples a subgraph $H_L$ by contracting each large hyperedge with probability $q$, and rejects samples that are contracted into a singleton. 
    The probability to keep a sampled $H_L$ is $u_L= u_{G_{large}}(q)$, so the probability mass for each $H$ is multiplied by $\frac{1}{u_L}$ compared to the distribution of contracting each large hyperedge with probability $q$.
    Second, the algorithm samples $H$ by contracting each small hyperedge with probability $q$ from $H_L$.
    
    Finally, a partial revelation step estimates $\ugp$ by $u_L\cdot u_H(p/q)$. We have
    \begin{align*}
    \E[u_H(p/q)] 
    &=\E_{H_L\sim G(q, 0)|D_L} [ \E_{H\sim H_L(0, q)} u_H(p/q) ] \\
    &=  \E_{H_L\sim G(q, 0)} \left[\frac{1}{u_L} 1[H_L\text{ disconnects}] \cdot \E_{H\sim H_L(0, q)} u_H(p/q) \right] \\
    &=  \frac{1}{u_L}\E_{H_L\sim G(q, 0)}[\E_{H\sim H_L(0, q)} u_H(p/q) ]\\
    &=\frac{1}{u_L}\E_{H\sim G(q)} u_H(p/q) = \frac{\ugp}{u_L}.
    \end{align*}
    The third equality drops the indicator function $1[H_L\text{ disconnects}]$ because $H_L$ being connected implies $u_H(p/q)=0$. The last step follows \Cref{lem:rc-unbias}.
\end{proof}

We now prove the inductive claim on the bias of the estimator.
\begin{lemma}\label{lem:alg2-bias}
    The algorithm outputs an estimator with negatively one-sided bias of at most $\delta$.
\end{lemma}
\begin{proof}
We prove by induction.
In the base case of $p^\lambda < 2^{-3N}$, the output is 0; so, the bias is negatively one-sided and upper bounded by $\ugp \le n^2p^\lambda \le (2^N)^2\cdot 2^{-3N} = 2^{-N}=\delta$.
The other base cases are unbiased.

In a random contraction step of universally small case, we take the average $X=\frac{1}{M}\sum_{i\le M}X_i$. By the inductive hypothesis, each $X_i$ satisfies $\E[X_i|H_i]-u_{H_i}(p/q) \in [-\delta, 0]$. Then,
\begin{align*}
\E[X]-\ugp &= \frac{1}{M}\sum_{i\le M}\E[X_i]-\ugp
= \frac{1}{M}\sum_{i\le M}\E_{H_i}[\E[X_i|H_i]]-\E_{H_i}[u_{H_i}(p/q)]\\
&= \frac{1}{M}\sum_{i\le M}\E_{H_i}[\E[X_i|H_i]-u_{H_i}(p/q)] \in [-\delta, 0].
\end{align*}

In the partial revelation case,
$Z$ is the DNF sampling estimator of $u_L = u_{\Glarge}(q)$, which is unbiased and independent of $X$. Next, we bound the bias of $X$ compared to $\ugp/u_L$.
We take average $X=\frac 1M \sum_{i\le M} X_i$, and each $X_i$ is an estimator for $u_{H_i}(p/q)$ with $\E[X_i|H_i]-u_{H_i}(p/q) \in [-\delta, 0]$ by the inductive hypothesis. 
Note that here $H_i$ is sampled from a different distribution where $\E[u_H(p/q)] = \ugp/ u_L$ by \Cref{lem:partial-revelation-bias}. Then,
\begin{align*}
\E[X]-\frac{\ugp}{u_L} &= \frac{1}{M}\sum_{i\le M}\E[X_i]-\frac{\ugp}{u_L}
=\frac{1}{M}\sum_{i\le M}\E_{H_i}[\E[X_i|H_i]]-\E_{H_i}[u_{H_i}(p/q)]\\
&= \frac{1}{M}\sum_{i\le M}\E_{H_i}[\E[X_i|H_i]-u_{H_i}(p/q)] \in [-\delta, 0].
\end{align*}
After scaling by $\E[Z]=u_L\le 1$, the overall bias of partial revelation case is
\[\E[XZ]-\ugp = \E[X]\E[Z]-\ugp = \left(\E[X]-\frac{\ugp}{u_L}\right)\cdot u_L \in [-\delta, 0].\]
\end{proof}

\subsection{Capped Relative Variance of the Estimator}
We first show that the recursive calls do not introduce relative variance.
Then we show that the full revelation base case outputs an unbiased estimator with bounded relative variance.
However, because the base case when $p^\lambda< 2^{-3N}$ is biased, we need to control $\delta$-capped relative variance instead of relative variance.
We conclude by bounding capped relative variance for the whole recursion in \Cref{lem:alg2-cap-relvar}.

In universally small hypergraphs, we have shown that a random contraction step has relative variance at most $n^2q^{-\lambda} = n^{12}$ in \Cref{lem:rc-relvar}.

\begin{fact}\label{fact:rejection-sampling}
    Suppose $X$ is a nonnegative random variable that takes value 0 outside a subspace $D$, and the measure of $D$ is $u_D$.
    If we do rejection sampling to sample $X'$ that only accepts samples in $D$, then $\eta[X']\le u_D \cdot \eta[X]$.
\end{fact}
\begin{proof}
    We have
    \[\E[X] =  u_D\E[X|D] + (1-u_D)\E[X|\bar{D}] = u_D \E[X|D] = u_D \E[X'].\]
    Then, $\E[X']=\E[X]/u_D$. By the same argument, $\E[X'^2] = \E[X^2] / u_D.$ 
    Therefore,
    \[\eta[X']+1=\frac{\E[X'^2]}{\E[X]^2} = \frac{\E[X^2]/u_D}{\E[X]^2/u_D^2} = u_D (\eta[X]+1), \]
    and $\eta[X']\le  u_D\cdot \eta[X]$.
\end{proof}

\begin{lemma}\label{lem:partial-revelation-relvar}
    In a partial revelation step, each estimator $u_{H_i}(p/q)$ has relative variance at most $n^4q^{-\beta}$.
\end{lemma}
\begin{proof}
    The samples $H_i$ can be viewed as generated by random contraction with survival probability $q$, but rejected when the contracted hyperedges in $\Elarge$ connects.
    The estimator $u_{H_i}(p/q)$ takes value 0 outside the accepting subspace (where $H$ is already contracted), and the accepting subspace has measure $u_L=u_{\Glarge}(q)$.
    So by \Cref{fact:rejection-sampling},
\[\eta[u_{H_i}(p/q)] \le  u_L \cdot \eta_{H\sim G(q)}[u_H(p/q)] \le  n^2q^{\lambda_L} \cdot n^2q^{-\lambda} \le n^4 q^{-\beta}.\]
Here, we used $u_L\le n^2q^{\lambda_L}$ by \Cref{lem:ugp-bound}, and $\eta_{H\sim G(q)}[u_H(p/q)]\le n^2q^{-\lambda}$ by \Cref{lem:rc-relvar}.

\end{proof}

\begin{lemma}\label{lem:full-revelation-relvar}
    The algorithm for full revelation case outputs an unbiased estimator of relative variance at most $3$.
\end{lemma}
\begin{proof}
    Unbiasedness is given by \Cref{lem:full-revelation-unbias}.
    Next, we bound the relative variance.
    
    In the full revelation case, each estimator $X_i$ is an indicator function of $H_i$ disconnecting, and the samples $H_i$ can be viewed as generated by random contraction with survival probability $p$, but rejected when the contracted hyperedges in $\Elarge$ connect.
    The estimator $X_i$ takes value 0 outside the accepting subspace (where $H$ is already contracted), and the accepting subspace has measure $u_L = u_{\Glarge}(p)$. 
    For indicator function $1_A$ on any event $A$, $\eta[1_A] < \E[1_A^2]/(\E[1_A]^2) = 1/\Pr[A]$.
    By \Cref{fact:rejection-sampling},
\[\eta[X_i] \le  u_L \cdot \eta_{H\sim G(p)}[1[H_i\text{ disconnects}]] \le u_L / \ugp \le  n^2p^{\lambda_L} / p^{\lambda} \le n^2 p^{-\beta}.\]
Here, we used $u_L\le n^2p^{\lambda_L}$ and $\ugp \ge p^\lambda$  by \Cref{lem:ugp-bound}.

By the definition of full revelation case,  $\beta<\lambda/N$, and $p^{-\beta} \le p^{-\lambda/N}$. Because we assumed $p^\lambda\ge 2^{-3N}$ (otherwise it would be a base case), $p^{-\beta}\le 2^{3}=8$. So, $\eta[X_i] \le 8n^2$.

$X$ is the average of $8n^2$ i.i.d.\ estimators $X_i$. So, $X$ has relative variance at most $1$ by \Cref{fact:cap-relvar-average}. $Z$ is independent of $X$ and has relative variance at most $1$ by \Cref{lem:dnf-graph}. Therefore, the product $XZ$ has relative variance $\le \eta[X]\eta[Z]+\eta[X]+\eta[Z] \le 3$ by \Cref{fact:cap-relvar-mult}.
\end{proof}

\begin{lemma}\label{lem:alg2-cap-relvar}
    The algorithm outputs a (biased) estimator $X$ of $\ugp$ with capped relative variance $\eta_\delta[X] \le 3$.
\end{lemma}
\begin{proof}
    We prove by induction on the recursion tree. (The induction is valid because each recursive call decreases $n$, $m$ or $\left\lceil m\ln \frac{1}{p}\right\rceil$, as argued earlier.)
    In the base case when $p^\lambda < 2^{-3N}$, the algorithm outputs $X=0$, whose $\delta$-capped relative variance is $0/\delta^2=0$.
    Other base cases are unbiased estimators of $\ugp$ with relative variance at most $3$ by \Cref{lem:brute-force,lem:full-revelation-relvar}.

    In the inductive step of random contraction, the first level estimator $u_{H_i}(p/q)$ has relative variance at most $n^2q^{-\lambda}-1 = n^{12}-1$ by \Cref{lem:rc-relvar}.
    The recursive estimator $X_i$ for $u_{H_i}(p/q)$ has negatively one-sided bias of $\le \delta$ by \Cref{lem:alg2-bias} and $\eta_\delta[X_i|H_i] \le 3$ by inductive hypothesis.
    By \Cref{lem:cap-relvar-concatenate},  (unconditionally) $\eta_\delta[X_i] \le 4\cdot n^{12} \cdot (3+1) = 16n^{12}$. The algorithm takes average of $16n^{12}$ estimators, so  the overall estimator has capped relative variance $\le 1$ by \Cref{fact:cap-relvar-average}.

    In the inductive step of the partial revelation case, the first level estimator $u_{H_i}(p/q)$ has relative variance $n^4q^{-\beta}$ by \Cref{lem:partial-revelation-relvar}.
    The recursive estimator $X_i$ for $u_{H_i}(p/q)$ has negatively one-sided bias of $\le \delta$ by \Cref{lem:alg2-bias} and $\eta_\delta[X_i|H_i]\le 3$ by inductive hypothesis.
    By \Cref{lem:cap-relvar-concatenate},  (unconditionally) $\eta_\delta[X_i] \le 16(n^4q^{-\beta}+1) \le 32n^{704}$. The algorithm takes average of $32n^{704}$ estimators to obtain $X$, so $\eta_\delta[X] \le 1$ by \Cref{fact:cap-relvar-average}.
    Finally, the algorithm multiplies $X$ by $Z$, which is an unbiased estimator of $u_{\Glarge}(p)$ with relative variance $\le 1$ by \Cref{lem:dnf-graph} and is independent of $X$. So, the product $XZ$ has $\delta$-capped relative variance $\le 3$ by \Cref{fact:cap-relvar-mult}.
\end{proof}

\subsection{Running Time}
The argument is similar to \Cref{sec:alg1-runtime}.
We color each recursive call as black or red. Intuitively, they represent a ``success'' or ``failure'' respectively.

For both the universally small and partial revelation cases, if the child subproblem is a phase node, then the recursive call is marked a success (i.e., a black node). This is the only type of success for the universally small case, which we call type~1 success.
For the partial revelation case, we have an additional situation where we declare type~2 success: when the parameter $\gamma$ decreases to $0.9\gamma$ and $|\ell-\lambda|\le 0.1\beta$.

Now, the recursion tree satisfies the following properties:
\begin{enumerate}
    \item Each subproblem makes $n^{O(1)}$ recursive calls.
    This is clear in the algorithm description.
    \item The algorithm reaches the base case after $O(\log n\cdot \log N)$ black recursive calls (interleaved with red recursive calls).
    We prove this in \Cref{lem:alg2-depth}.
    \item At each subproblem, the expected number of red recursive calls is $o(1)$.
    We prove this later in the section.
\end{enumerate}
\Cref{lem:recursive-tree-size-bound} shows that these properties give a upper bound of $n^{O(\log n\cdot \log N)}$ on the number of recursive calls. 
If we charge the time of DNF sampling and random contraction to the subproblem on the contracted hypergraph, then each subproblem spends $O(n^2m)$ time outside the recursive calls, where the bottlenecks are DNF sampling and DNF probability estimation given by \Cref{lem:dnf-graph}. Therefore, the overall expected running time is $m\cdot n^{O(\log n\cdot \log N)}$.

\begin{lemma}\label{lem:alg2-depth}
    There can be at most $O(\log n \cdot \log N)$ black recursive calls from root to a base case.    
\end{lemma}
\begin{proof}
    There are at most $O(\log n)$ phases from root to a base case because each phase node decreases $n$ to at most $0.8n$ compared to the last phase node.
    
    Within each phase, $\gamma$ is non-increasing by \Cref{lem:gamma-monotone}, and each type~2 success decreases $\gamma$ to at most $0.9\gamma$. 
    Note that whenever the type~2 success happens, we are in the case that $|\gamma-\beta| = |\ell - \lambda| \le  0.1\beta$.
    Initially, $\beta \le \lambda$ and the algorithm reaches a base case when $\beta<\lambda/N$; so, there can be at most $O(\log N)$ recursive calls with the type~2 success in a branch of the recursion tree within a phase.
\end{proof}

Next, we establish property 3 for each case of the recursive step, i.e., that the expected number of red recursive calls is $o(1)$.

\paragraph{Universally small case.}
In the universally small case, we can use the same analysis as \Cref{alg:enumeration}.
The algorithm runs $O(n^{12})$ recursive calls,
while the failure probability of each call is $n^{2.7}q^{1.5\lambda}$ by \Cref{cor:rc-0.5n-bound}.
The algorithm chooses $q^\lambda = n^{-10}$ so that their product is $o(1)$.

\paragraph{Partial revelation case.}
The algorithm runs $O(n^{704})$ recursive calls. We next prove that the failure probability is at most $n^2q^{1.01\beta} = n^{-705}$, so that their product is $o(1)$.

Consider a computation node $v$ in the partial revelation case. Let $u$ be the phase ancestor of $v$ (or $u=v$ if $v$ is a phase node), and let $w$ be a child subproblem called by $v$.
Consider the recursive call from $v$ to $w$.
There are two types of success: (1) $w$ is a phase node, i.e.\ the recursive call decreases the number of vertices by a constant factor compared to the phase ancestor $u$, and
(2) the parameter $\gamma_w$ defined in $w$ is at most $0.9\gamma$ for $\gamma$ defined in $v$, and we are in the case that $|\gamma-\beta|\le 0.1\beta$.

Recall that $\gamma=\ell-\lambda_L$ and $\beta=\lambda-\lambda_L$, where $\lambda_L$ is the min-cut value in $\Glarge$, which is equal to the minimum degree cut value in $\Glarge$ by \Cref{lem:all-large-degree-cuts}.

\begin{lemma}\label{lemma:failure-probability}
    The failure probability in each partial revelation case is at most $n^2q^{1.01\beta}$.
\end{lemma}
\begin{proof}
    Let $v$ be the current computation node, $u$ be its phase ancestor (or $u=v$ if $v$ is a phase node), and $w$ be a child subproblem called by $v$.
    All parameters and notations are defined in $v$ by default, and we will add subscripts if they are defined in $u$ or $w$.
    
    We consider three cases.
    The first case is $\ell > \lambda + 0.1\beta$.
    Then the probability to not contract any large hyperedge (which is necessary for failure) is $q^\ell / u_L \le q^{\lambda+0.1\beta}/q^{\lambda_L} = q^{1.1\beta}$.

    The second case is $\ell < \lambda - 0.1\beta$.
    Failure happens when the following two independent events both happen:
      \begin{enumerate}
      \item No large hyperedge is contracted.\label{item:failure-event-1}
      \item The random contraction in $G_{small}$ does not decrease the number of vertices to $\le 0.8n$.\label{item:failure-event-2}
      \end{enumerate}
    The first event happens with probability $q^{\ell} / u_L\le q^{\lambda-0.1\beta-\lambda_L} = q^{0.9\beta}$.
    For the second event, we will apply \Cref{lem:rc-small-rank-size-bound} on $\Gsmall$.
    The min-cut value in $\Gsmall$ is at least $\lambda-\ell$ because the union of the min-cut in $\Gsmall$ and all hyperedges in $\Elarge$ is a cut in $G$.
    The maximum rank in $\Gsmall$ is at most $0.7n$ by \Cref{fact:elarge-rank-range}.
    So, we can apply \Cref{lem:rc-small-rank-size-bound} with $A=\frac{10}{7}, B=\frac{8}{7}$ to bound the probability of the second event by $n^2 q^{1.1(\lambda-\ell)} \le n^2q^{1.1(0.1\beta)} = n^2q^{0.11\beta}$.
    The total failure probability is upper bounded by $q^{0.9\beta} \cdot n^2 q^{0.11\beta} = n^2q^{1.01\beta}$.

    The last case is $\lambda-0.1\beta \le \ell \le \lambda+0.1\beta$. We will upper bound the probability that failure event~\ref{item:failure-event-1} and the following new failure event both happen:
    \begin{enumerate}\setcounter{enumi}{2}
        \item $\gamma_w > 0.9\gamma$. Hence,
    \[\ell-\lambda_L(w)> 0.9(\ell - \lambda_L) \ge 0.9(\lambda - 0.1\beta - \lambda_L) = 0.9(0.9\beta) > 0.8\beta.    \]
    This condition becomes $\lambda_L(w) < \ell-0.8\beta$ which means that some vertex $s\in G_w$ is incident to less than $\ell-0.8\beta \le \lambda-0.7\beta$ hyperedges in $\Elarge$.
    \label{item:failure-event-3}
    \end{enumerate}
    For failure event~\ref{item:failure-event-1}, no large hyperedge gets contracted. This happens with probability at most $q^\ell / u_L\le q^{\ell-\lambda_L}$.
    
    We now bound failure event~\ref{item:failure-event-3} conditioned on failure event~\ref{item:failure-event-1} occurring. Here, it suffices to upper bound the probability that some vertex $s\in G_w$ is incident to less than $\lambda-0.7\beta$ hyperedges in $\Elarge$, given that no large hyperedge is contracted. The hypergraph under random contraction is now $\Gsmall$, and this random contraction is independent of the previous event on $\Elarge$.

    

    
    Consider the exponential contraction process on $\Gsmall$.
    Replace each hyperedge with two copies (these copies are ``half-hyperedges'' with survival probability $\sqrt{q}$ instead of $q$).
    Assign a head to each hyperedge, in a way that two copies from a hyperedge have different heads.
    The orientation may change during the process after each hyperedge arrives.
    The orientation is consistent, which means for any fixed subgraph we always choose the same orientation. Besides this requirement, the orientation is arbitrary.
    We assign a representative to each contracted supervertex during the process. Originally, each vertex is its own representative. When a hyperedge $e$ is contracted, the representative of the head of $e$ becomes the representative of the new contracted supervertex.
    
    Define the critical hyperedges for a supervertex $s$ to be the hyperedges that contain $s$ as a tail.
    For the failure event~\ref{item:failure-event-3} to happen, there exists a supervertex $\ts$ in $w$ that is incident to less than $\lambda-0.7\beta$ hyperedges in $\Elarge$.
    Let $s$ be the representative of $\ts$, so $s$ is a vertex in node $v$.
    By extension, we also denote by $\ts$ the supervertices that contain $s$ as a representative throughout the exponential contraction process. (Note that initially $\ts = s$.)
    Since we do not contract any large hyperedge, the number of large hyperedges that contain $\ts$ can only increase over time, but by assumption, $\ts$ is incident to less than $\lambda - 0.7 \beta$ large hyperedges at the end. So throughout the process, $\ts$ is always incident to less than $\lambda - 0.7 \beta$ hyperedges in $\Elarge$.
    Because the degree cut of $\ts$ has value at least $\lambda$, $\ts$ is always incident to at least $0.7\beta$ hyperedges in $\Esmall$. Then, it always has $0.7\beta$ critical (half-)hyperedges after duplication.
    
    By \Cref{lem:critical-edge-not-arrive-prob}, 
    the probability that such a  vertex  $s$ will survive as a representative until time $\ln \frac{1}{\sqrt{q}}$ is at most $(\sqrt{q})^{0.7\beta} = q^{0.35\beta}$.
    By a union bound, the probability of failure event~\ref{item:failure-event-3} (conditioned on failure event~\ref{item:failure-event-1}) is at most $nq^{0.35\beta}$.
    The total failure probability is $q^{0.9\beta}\cdot nq^{0.35\beta}=nq^{1.25\beta}$.
\end{proof}

\section{Conclusion}
\label{sec:conclusion}
In this paper, we initiated the study of unreliability in hypergraphs and provided quasi-polynomial time approximation schemes for the problem. The immediate open question is whether there is a PTAS (or even FPTAS) for this problem. More generally, we hope that our work will inspire further exploration of the rich space of reliability problems in hypergraphs. For instance, a natural complementary question to network unreliability is that of network reliability, i.e., estimating the probability that a network stays connected under independent random failures of the (hyper)edges. For graphs, there are PTAS for this problem using very a different set of techniques from network unreliability~\cite{GuoJ19,GuoH20,ChenGZZ23}; the corresponding problem in hypergraphs remains open.


\section*{Acknowledgments}
RC and DP were supported in part by NSF grants CCF-1750140 (CAREER), CCF-1955703, and CCF-2329230. This research was done at the Simons Institute, UC Berkeley in Fall 2023 when RC was a Visiting Graduate Student, JL was a Simons Fellow, and DP was a Visiting Scientist in the semester program on {\em Data Structures and Optimization for Fast Algorithms} at the institute. DP also wishes to acknowledge the support of Google Research, where he was a (part-time) visiting faculty researcher at the time of this research. RC and DP would also like to thank William He and Davidson Zhu for discussions at an early stage of this research.

\bibliographystyle{alpha}
\bibliography{ref}

\appendix

\section{Missing Proofs from \Cref{sec:prelim}}
\label{sec:proofs}

\paragraph{Proofs on capped relative variance.}
We first state the following standard facts about relative variance.
The also serve as special cases of \Cref{lem:cap-relvar-concatenate} and~\Cref{lem:cap-relvar-approx}.
\begin{fact}[\cite{Karger17}]\label{fact:relvar-compose}
    Suppose $Y$ is an unbiased estimator of $x$ with relative variance $\eta_1$, and given any fixed value of $Y$, $Z$ is an unbiased estimator of $Y$ with relative variance $\eta_2$. Then $Z$ is an unbiased estimator of $x$ with relative variance $(\eta_1+1)(\eta_2+1)-1$. 
\end{fact}
\begin{lemma}[\cite{Karger17}]\label{lem:relvar-approx}
    Fix $\eps, \delta\in(0,1)$. For a random variable $X$ with relative variance $\eta[X]$, the median of $O(\log\frac{1}{\delta})$ averages of $O(\frac{\eta[X]}{\eps^2})$ independent samples of $X$ is a $(1\pm \eps)$-approximation of $\E[X]$ with probability at least $1-\delta$.
\end{lemma}

Recall that all random variables discussed in the paper are non-negative. This will be implicitly used in the following proofs.

\begin{proof}[Proof of \Cref{fact:cap-relvar-average}]
    For i.i.d.\ samples, the denominator $\max\{(\E[X])^2, \delta^2\}$ is identical, while the numerator $\var[X]$ is divided by $M$ after taking average. So, the capped relative variance is divided by $M$.
\end{proof}

\begin{proof}[Proof of \Cref{lem:cap-relvar-concatenate}]
The special case is given by \Cref{fact:relvar-compose}. Next, we consider general case $\delta>0$.
We start with
\[\mu=\E[Z] = \E_Y[\E[Z|Y]] \in [\E[Y-\delta],\E[Y]] = [x-\delta, x].\]
This implies $\mu+\delta\ge x$ and $\max\{\mu,\delta\}\ge \frac 12 x$.
We can bound $\eta_\delta[Z]$ by
\[\eta_\delta[Z] = \frac{\E[Z^2]-\mu^2}{\max\{\mu^2, \delta^2\}} \le \frac{\E[Z^2]}{\max\{\frac 14 x^2, \delta^2\} }.\]
Notice that
\[\E[Z^2] = \E_Y[\E[Z^2|Y]]
=\E_Y\left[\max\{Y^2,\delta^2\} \cdot \frac{\E[Z^2|Y]}{\max\{Y^2,\delta^2\}}\right]
\le \max\{\E[Y^2],\delta^2\}\cdot \max_Y \frac{\E[Z^2|Y]}{\max\{Y^2, \delta^2\}}.
\]
Thus,
\[\eta_\delta[Z] \le \frac{\max\{\E[Y^2],\delta^2\}}{\max\{\frac 14 x^2, \delta^2\}} \cdot \max_Y \frac{\E[Z^2|Y]}{\max\{Y^2, \delta^2\}}
\le \max\left\{\frac{4\E[Y^2]}{x^2}, 1\right\}\cdot \max_Y \frac{\E[Z^2|Y]}{\max\{Y^2, \delta^2\}}.
\]

The first term is upper bounded by $4(\eta[Y]+1)$ because $\eta[Y]+1 = \frac{\E[Y^2]}{x^2}$ by definition.
The second term is upper bounded by $h+1$ because for all fixed $Y$,
\[ h \ge \eta_\delta[Z|Y] = \frac{\E[Z^2|Y]-(\E[Z|Y])^2}{\max\{\E[Z|Y], \delta\}^2}
\ge \frac{\E[Z^2|Y]-Y^2}{\max\{Y, \delta\}^2}
\ge \frac{\E[Z^2|Y]}{\max\{Y, \delta\}^2}-1.\]
In conclusion $\eta_\delta[Z] \le 4\cdot (\eta[Y]+1) \cdot (h+1)$.
\end{proof}
\eat{
\begin{proof}[Proof of \Cref{lem:cap-relvar-concatenate}]
The special case is given by \Cref{fact:relvar-compose}. Next, we consider general case $\delta>0$.
We start with
\[\mu=\E[Z] = \E_Y[\E[Z|Y]] \in [\E[Y-\delta],\E[Y]] = [x-\delta, x].\]
When $x>2\delta$, we have $\mu\ge x-\delta\ge \delta$ and $\mu\ge x-\frac x2 x = \frac x2$. Then,
\[\eta_\delta[Z] \le \frac{\E[Z^2]-\mu^2}{\mu^2} \le \frac{4\E[Z^2]}{x^2}.\]
When $x\le 2\delta$, we also have
\[\eta_\delta[Z] \le \frac{\E[Z^2]-\mu^2}{\delta^2} \le \frac{4\E[Z^2]}{x^2}.\]
We can bound the last term by
\begin{align*}
   \frac{4\E[Z^2]}{x^2}
    &= \frac{4}{x^2}\E_Y[\E[Z^2|Y]]
    = \frac{4}{x^2}\E_Y\left[\frac{\E[Z^2|Y]}{Y^2}\cdot Y^2\right]
    \le \frac{\E[Y^2]}{x^2}\cdot \max_Y \frac{4\cdot \E[Z^2|Y]}{Y^2}
\end{align*}
We have $\eta[Y]+1 = \frac{\E[Y^2]}{x^2}$, and for all fixed $Y$,
\[ h \ge \eta_\delta[Z|Y] = \frac{\E[Z^2|Y]-(\E[Z|Y])^2}{\max\{\E[Z|Y], \delta\}^2}
\ge \frac{\E[Z^2|Y]-Y^2}{\max\{Y, \delta\}^2}
\ge \frac{\E[Z^2|Y]-Y^2}{Y^2} = \frac{\E[Z^2|Y]}{Y^2}-1,\]
so the statement holds.
\end{proof}
}

\begin{proof}[Proof of \Cref{fact:cap-relvar-mult}]
Under the assumptions, \begin{align*}\max\{(\E[X])^2, \delta^2\}\cdot \max\{(\E[Z])^2, \delta^2\} &= \max\{(\E[X]\E[Z])^2, (\E[X])^2\delta^2, (\E[Z])^2 \delta^2, \delta^4\}\\&\le \max\{(\E[X]\E[Z])^2, \delta^2\}.\end{align*}
Therefore,
\begin{align*}
    \eta_\delta[XZ] 
    &= \frac{\var[XZ]}{\max\{(\E[XZ])^2, \delta^2\}} 
    = \frac{\var[X]\var[Z] + \var[X](\E[Z])^2+\var[Z](\E[X])^2}{\max\{(\E[X])^2(\E[Z])^2, \delta^2\}} \\
    &\le \frac{\var[X]\var[Z] + \var[X](\E[Z])^2+\var[Z](\E[X])^2}{\max\{(\E[X])^2, \delta^2\}\cdot \max\{(\E[Z])^2, \delta^2\}} \\
    &\le \eta_\delta[X]\eta_\delta[Z] + \eta_\delta[X] + \eta_\delta[Z]. \qedhere
\end{align*}
\end{proof}

\begin{proof}[Proof of \Cref{lem:cap-relvar-approx}]
If $\E[X]\ge \delta$, then $\eta_\delta[X] = \eta[X]$ and the lemma follows \Cref{lem:relvar-approx}.

If $\E[X]< \delta$, $\eta_\delta[X] = \var[X]/\delta^2$.
Consider $X'=X+\delta$.
\[\eta[X'] = \frac{\E[X'^2]}{\E[X']^2}-1 =\frac{\E[X^2] + 2\delta\E[X]+\delta^2}{\E[X]^2 + 2\delta \E[X]+\delta^2}-1 \le \frac{\var[X]}{\delta^2} = \eta_\delta[X]
\]
So we can apply \Cref{lem:relvar-approx} to get a $(1+\eps)$-approximation of $X+\delta$, which is a $(1+\eps, \delta)$-approximation of $X$ after subtracting $\delta$.
\end{proof}

\begin{proof}[Proof of \Cref{fact:relvar-linear-combination}]
    Suppose $X=\sum_{i\le k}\alpha_iX_i$ and $\E[X_i] = \mu_i$. Then
    \begin{align*}
    \eta[X] &= \frac{\E[(\sum_i\alpha_iX_i)^2]}{(\sum_i\alpha_i\mu_i)^2}-1
    =\frac{\sum_{i,j}\alpha_i\alpha_j\E[X_iX_j]}{\sum_{i,j}\alpha_i\alpha_j\mu_i\mu_j} - 1\\
    &= \frac{\sum_i\alpha_i^2\E[X_i^2] + \sum_{i\ne j}\alpha_i\alpha_j \E[X_i]\E[X_j] }{\sum_i\alpha_i^2\mu_i^2 + \sum_{i\ne j}\alpha_i\alpha_j\mu_i\mu_j} -1\\
    &= \frac{\sum_i\alpha_i^2(\E[X_i^2] - \mu_i^2) }{\sum_i\alpha_i^2\mu_i^2 + \sum_{i\ne j}\alpha_i\alpha_j\mu_i\mu_j} \\ 
    & \le \frac{\sum_i\alpha_i^2(\E[X_i^2]-\mu_i^2)}{\sum_i\alpha_i^2\mu_i^2}\\
    &\le \max_i \frac{\E[X_i^2]-\mu_i^2}{\mu_i^2} = \max_i \eta[X_i]
    \end{align*}
\end{proof}

\begin{proof}[Proof of \Cref{fact:relvar-max}]
    \[\eta[X] = \frac{\E[X^2]}{(\E[X])^2} -1 \le \frac{\E[X\cdot M]}{(\E[X])^2}-1 = \frac{\E[X]\cdot M}{(\E[X])^2}-1 =\frac{M}{\E[X]}-1\]
\end{proof}

\begin{proof}[Proof of \Cref{lem:mc-relvar}]
The estimator $X$ follows a binomial distribution of parameter $p_D$, so
$\E[X]=p_D$ and $\var[X]=p_D(1-p_D) \le p_D$.
It follows that $\eta[X] = \frac{\var[X]}{\E[X]^2}\le \frac{1}{p_D}$.
For capped relative variance,
when $p_D \ge \delta$, $\eta_\delta[X] = \frac{\var[X]}{p_D^2}\le \frac{1}{p_D}$;
when $p_D < \delta$,
$\eta_\delta[X] = \frac{\var[X]}{\delta^2}\le \frac{p_D}{\delta^2} \le \frac{1}{\delta}$.
\end{proof}

We prove \Cref{lem:dnf-unbias,lem:dnf-sampling} together:

\begin{proof}[Proof of \Cref{lem:dnf-unbias,lem:dnf-sampling}]
    For each clause $C_i$, suppose it has $a_i$ positive literals and $b_i$ negative literals. Then the probability that $C_i$ is satisfied is $p^{a_i}(1-p)^{b_i}$. Denote this as $u_i$.

    We first give a sampling algorithm as follows.
    Pick a clause $C_i$ with probability $u_i/U$, where $U=\sum_{C_j\in F}u_j$ is a normalizer. Set all variables in $C_i$ to satisfy $C_i$, i.e., variables in positive literals are \textsc{True} and variables in negative literals are \textsc{False}. Set other variables randomly, such that each variable is \textsc{True} with probability $p$ independently.
    Denote the Boolean vector of variable values by $x$.
    The algorithm accepts a sample with probability $\frac{1}{f(x)}$ and rejects it with probability $1-\frac{1}{f(x)}$, where $f(x)$ the number of clauses satisfied by the sample values $x$.
    Note that $x$ at least satisfies $C_i$; so $f(x) \ge 1$, i.e.\ the algorithm only draw samples that satisfy $F$.

    Consider any $x$ that satisfies $F$. Suppose there are $a_x$ \textsc{True}s and $b_x$ \textsc{False}s in $x$.
    Let the clauses satisfied by $x$ be $C_{i_1}, C_{i_2}, \ldots, C_{i_{f(x)}}$. In the algorithm, $x$ can be drawn in $f(x)$ different ways when we pick each $C_{i_k}$. The total probability to sample $x$ is
    \[\sum_{k=1}^{f(x)} \frac{u_{i_k}}{U} \cdot p^{a_x-a_{i_k}}(1-p)^{b_x-b_{i_k}}
    = \sum_{k=1}^{f(x)} \frac{p^{a_{i_k}}(1-p)^{b_{i_k}}}{U} \cdot p^{a_x-a_{i_k}}(1-p)^{b_x-b_{i_k}} = \frac{f(x)p^{a_x}(1-p)^{b_x}}{U}\]
    We accept $x$ with probability $\frac{1}{f(x)}$, so the probability to sample and accept $x$ is $p^{a_x}(1-p)^{b_x}/U$.
    The total the probability that the algorithm accepts a sample is
    \[\sum_{x:f(x)\ge 1}\frac{p^{a_x}(1-p)^{b_x}}{U} = \frac{u_F(p)}{U}\]
    Therefore, among the accepted samples, $x$ is sampled with probability $p^{a_x}(1-p)^{b_x}/u_F(p)$, which is the desired distribution.

    The algorithm takes $O(N)$ time to pick $C_i$ and sample $x$. Then, the algorithm takes $O(NM)$ time to count $f(x)$ by checking whether each clause is satisfied by $x$. In addition, the algorithm accepts a sample with probability $\frac{1}{f(x)} \ge \frac{1}{M}$. So, we expect to accept a sample in $O(M)$ trials, and the expected running time is $O(NM^2)$ in total.

    Next, we show that this algorithm can be used to estimate $u_F(p)$.
    We proved that the algorithm accepts a sample with probability $u_F(p)/U$, while $U$ can be directly computed.
    So, we can define an estimator $X$ by $X=U$ when the algorithm accepts a sample, and $X=0$ when rejects. Then $\E[X]=u_F(p), \var[X] < \E[X^2] = u_F(p) \cdot U$.
    Notice that $u_F(p)\ge u_j$ for any $j$. So $U=\sum_{C_j\in F} u_j\le M\cdot u_F(p)$, and $
    \eta[X]=\frac{\var[X]}{\E[X]^2} \le \frac{U}{u_F(p)} \le M$.
    By taking the average of $M$ estimators, we get an unbiased estimator of relative variance at most $1$.
    The running time is again $O(NM^2)$.
\end{proof}


\end{document}